\documentclass[11pt]{article}
\usepackage[utf8]{inputenc}
\usepackage{fullpage}
\usepackage{amsmath,amsthm,amsfonts,mathtools}
\usepackage{changepage}
\usepackage{tikz}
\usepackage{paralist}
\usepackage[shortlabels]{enumitem}
  \setlist[itemize]{leftmargin=*}
  \setlist[enumerate]{leftmargin=*}
\usepackage{hyperref}
\usepackage[capitalize,nameinlink]{cleveref}
\usepackage[linesnumbered,boxed,noend,noline]{algorithm2e}
\usepackage{framed}
\usepackage{algpseudocode}
\usetikzlibrary{positioning,matrix,calc,arrows,shapes,fit,decorations,angles,quotes}
\usepackage{thmtools}
\usepackage{thm-restate}
\usepackage[nottoc,notlot,notlof]{tocbibind}
\usepackage{tcolorbox}

\makeatletter
\newcommand{\customlabel}[2]{%
   \protected@write \@auxout {}{\string \newlabel {#1}{{#2}{\thepage}{#2}{#1}{}} }%
   \hypertarget{#1}{#2}
}
\makeatother

\definecolor{shadecolor}{gray}{0.9}

\newcommand{\eps}{\varepsilon}
\newcommand{\N}{\mathbb{N}}
\newcommand{\E}{\mathbb{E}}
\newcommand{\F}{\mathbb{F}}
\newcommand{\R}{\mathbb{R}}
\newcommand{\C}{\mathsf{C}}
\newcommand{\rk}{\mathsf{rk}}

\newcommand{\seq}{\subseteq}

\newcommand{\poly}{\mathrm{poly}}
\newcommand{\ip}[1]{\left\langle #1 \right\rangle}

\newcommand{\vecx}{\vec{x}}
\newcommand{\vecy}{\vec{y}}
\newcommand{\vecz}{\vec{z}}
\newcommand{\vecw}{\vec{w}}
\renewcommand{\sp}[1]{{\rm span}\{#1\}}

\newcommand{\norm}[1]{{\|#1\|}}
\newcommand{\abs}[1]{{\left|#1\right|}}

\renewcommand{\sp}{\mathsf{span}}
\newcommand{\ALG}{\mathsf{ALG}}
\newcommand{\DS}{\mathsf{DS}}

\newcommand{\Spec}{{\mathrm{Spec}}}

\newcommand{\1}{\mathbf{1}}
\newcommand{\RM}{\sf{RM}}

\newcommand{\DSParams}[4]{{\sf{DS} \left[
{\small
\begin{array}{r  l}
\textsf{preprocessing time:}&\enspace {#1} \\
\textsf{memory used:}&\enspace {#2} \\
\textsf{query time:}&\enspace {#3} \\
\textsf{success rate:}&\enspace {#4} \\
\end{array}
}
\right]}}

\newtheorem{theorem}{Theorem}[section]
\newtheorem{lemma}[theorem]{Lemma}
\newtheorem{proposition}[theorem]{Proposition}
\newtheorem{claim}[theorem]{Claim}

\newtheorem{definition}[theorem]{Definition}

\newtheorem{remark}[theorem]{Remark}

\title{Worst-Case to Average-Case Reductions\\ via Additive Combinatorics}
\author{
Vahid R. Asadi\thanks{University of Waterloo. Email: \texttt{vrasadi@uwaterloo.ca}.}
\and
Alexander Golovnev\thanks{Georgetown University. Email: \texttt{alexgolovnev@gmail.com}.}
\and
Tom Gur\thanks{University of Warwick. Email: \texttt{tom.gur@warwick.ac.uk}. Tom Gur is supported by the UKRI Future Leaders Fellowship MR/S031545/1.}
\and
Igor Shinkar\thanks{Simon Fraser University. Email: \texttt{ishinkar@sfu.ca}.}
}
\date{}
\begin{document}

\maketitle

\begin{abstract}
We present a new framework for designing worst-case to average-case reductions. For a large class of problems, it provides an explicit transformation of algorithms running in time~$T$ that are only correct on a small (subconstant) fraction of their inputs into algorithms running in time~$\widetilde{O}(T)$ that are correct \emph{on all} inputs. 

Using our framework, we obtain such efficient worst-case to average-case reductions for fundamental problems in a variety of computational models; namely, algorithms for matrix multiplication, streaming algorithms for the online matrix-vector multiplication problem, and static data structures for all linear problems as well as for the multivariate polynomial evaluation problem.

Our techniques crucially rely on additive combinatorics. In particular, we show a local correction lemma that relies on a new probabilistic version of the quasi-polynomial Bogolyubov-Ruzsa lemma.
\end{abstract}

\newpage

\setcounter{tocdepth}{3}
\thispagestyle{empty}
\newpage
\tableofcontents
\newpage
\pagenumbering{arabic}

\section{Introduction}
Worst-case to average-case reductions provide a method for transforming algorithms that can only solve a problem for a fraction of the inputs into algorithms that can solve the problem \emph{for all} inputs.

For instance, consider one of the most fundamental algorithmic problems: matrix multiplication. Suppose we have an average-case algorithm $\ALG$ that can correctly compute the product $A \cdot B$ on an $\alpha$-fraction of matrices $A,B\in\F^{n\times n}$; that is, $ \Pr[\ALG(A,B)=A\cdot B] \geq \alpha$. Is it possible to use $\ALG$ to obtain an algorithm that computes $A \cdot B$ \emph{for all} input matrices?
A worst-case to average-case reduction will give a positive answer to this question, boosting the \emph{success rate} $\alpha$ to $1$, without incurring significant overhead. Of course, the same question can be asked with respect to any other computational problem.

In this paper, we study such reductions for average-case algorithms where the success rate $\alpha$ could be very small, such as in the $\%1$ regime, and even when $\alpha$ tends to zero rapidly (i.e., for algorithms that are only correct on a vanishing fraction of their inputs). 
There are two natural perspectives in which we can view such reductions. On the one hand, they can provide a proof that a problem retains its hardness even in the average case. On the other hand, they provide a paradigm for designing worst-case algorithms, by first constructing algorithms that are only required to succeed on a small fraction of their inputs, and then using the reduction to obtain algorithms that are correct on all inputs.

\paragraph{Background and context.}
The study of the average-case complexity originates in the work of Levin~\cite{L86}. A long line of works established various barriers to designing worst-case to average-case reductions for $\mathbf{NP}$-complete problems (see, e.g., \cite{I11} and references therein). We refer the reader to the classical surveys by Impagliazzo~\cite{I95}, and Bogdanov and Trevisan~\cite{BT06} on this topic. 

On the positive side, Lipton~\cite{L91} proved that the matrix permanent problem admits a polynomial-time worst-case to average-case reduction. Ajtai~\cite{A1996} designed worst-case to average-case reductions for certain lattice problems, which led to constructions of efficient cryptographic primitives from worst-case assumptions~\cite{AD97,R04}. Other number-theoretic problems in cryptography have been long known to admit such reductions due to random self-reducibility: the discrete logarithm problem, the RSA problem, and the quadratic residuosity problem (see, e.g., \cite{S09}). For the matrix multiplication problem, there is a weak reduction that requires the average-case algorithm to succeed with very high probability $3/4$ (see~\cref{tech:challange}). There are also known worst-case to average-case reductions for many problems that are not thought to be in $\mathbf{NP}$~\cite{FF93,BFNW93,STV01}.

Recently, the study of fine-grained complexity~\cite{V18} of algorithmic problems sparked interest in designing \emph{efficient} worst-case to average-case reductions for such problems as orthogonal vectors, 3SUM, online matrix-vector multiplication, $k$-clique, and others. Such reductions are motivated by fine-grained cryptographic applications. A large body of work is devoted to establishing fine-grained worst-case to average-case reductions for the $k$-clique problem, orthogonal vectors, 3SUM, and various algebraic problems, as well as building certain cryptographic primitives from them~\cite{BRSV17, BRSV18, GR18, LLW19, BBB21,DLV20}. Since there are  no known constructions of one-way functions and public-key cryptography from well-established fine-grained assumptions, the question of constructing efficient worst-case to average-case reductions for other fine-grained problems still attracts much attention.

\subsection{Our contribution}
We design a framework for showing explicit worst-case to average-case reductions, and we use it to obtain reductions for fundamental problems in a variety of computational models. Informally, we show that if a problem has an algorithm that runs in time~$T$ and succeeds on $\alpha$-fraction of its inputs (even for sub-constant success rate $\alpha$), then there exists a worst-case algorithm for this problem, which runs in time $\widetilde{O}(T)$. We design such reductions for the matrix multiplication problem in the setting of algorithms, for the online matrix-vector multiplication problem in the streaming setting, for \emph{all} linear problems in the setting of static data structures, and for the problem of multivariate polynomial evaluation. We describe these results in detail below.

\subsubsection{Algorithms for matrix multiplication}
Recall that in the matrix multiplication problem, the goal is simply to compute the product of two given matrices $A,B\in\F^{n\times n}$. A long line of research, culminating in the work of Alman and Vassilevska Williams~\cite{AV21}, led to matrix multiplication algorithms performing $O(n^{2.37286})$ operations. We present a worst-case to average-case reduction for the matrix multiplication problem over prime fields.
Namely, we show that if there exists a (randomized) algorithm that, given two matrices $A,B \in \F^{n \times n}$, runs in time $T(n)$
and correctly computes their product for a small fraction of all possible inputs, then
there exists a (randomized) algorithm that runs in $\widetilde{O}(T(n))$ time and outputs the correct answer for all inputs.
Formally, we have the following theorem.

\begin{restatable}{mtheorem}{mmreduction}\label{thm:mm-alg-reduction-intro}
    Let $\F = \F_p$ be a prime field, $n \in \N$, and $\alpha\coloneqq\alpha(n) \in (0,1]$.
    Suppose that there exists an algorithm $\ALG$ that, on input two matrices $A, B \in \F^{n \times n}$ runs in time $T(n)$
    and satisfies
    \begin{equation*}
        \Pr[\ALG(A,B)=A\cdot B] \geq \alpha \,,
    \end{equation*}
    where the probability is taken over the random inputs $A,B\in \F^{n \times n}$ and the randomness of~$\ALG$.
    \begin{itemize}
            \item If $\abs{\F} \leq 2/\alpha$, then there exists a randomized algorithm $\ALG'$ that for \emph{every} input $A, B \in \F^{n \times n}$ and $\delta>0$, runs in time $\frac{\exp(O(\log^5(1/\alpha)))}{\delta} \cdot T(n)$ and outputs $AB$ with probability at least $1-\delta$.
        \item If $\abs{\F} \geq 2/\alpha$, then there exists a randomized algorithm $\ALG'$ that for \emph{every} input $A, B \in \F^{n \times n}$ and $\delta>0$,
        runs in time $O(\frac{1}{\delta \cdot \alpha^4} \cdot T(n))$ and outputs $AB$ with probability at least $1-\delta$.
    \end{itemize}
\end{restatable}

For example, if we have an algorithm that succeeds on $\alpha$ fraction of the inputs for $\alpha > \exp(-\sqrt[6]{\log(n)})$ in time $T(n) = n^c$, then we get an algorithm that works for all inputs and runs in time $n^{c + o(1)}$. In particular, if we have an $n^{2+o(1)}$ algorithm that succeeds on $\alpha > \exp(-\sqrt[6]{\log(n)})$ fraction of the inputs, then there is a worst case algorithm with running time $n^{2+o(1)}$.

\subsubsection{Data structures for all linear problems}
The class of linear problems plays a central role throughout computer science and mathematics, yielding a myriad of applications both in theory and practice. Our next contribution gives worst-case to average-case reductions for static data structures for \emph{all linear problems}. Recall that a linear problem $L_A$ over a field $\F$ is defined by a matrix $A \in \F^{m \times n}$.\footnote{Formally, $L_A$ is defined by an infinite sequence of matrices $(A_n)_{n\geq 1}$, where $A_n\in\F^{m \times n}$ for $m=m(n)$.} An input to the problem is a vector $v \in \F^n$, which is preprocessed into $s$ memory cells. Then, given a query $i \in [m]$, the goal is to output $\ip{A_i, v}$, where $A_i$ is the $i$'th row of $A$, by probing at most $t$ of the memory cells, where $t$ is called the query time. 

Note that the trivial solutions for data structure problems are to either:
\begin{inparaenum}[(i)]
\item store only $s=n$ memory cells containing the input $v$, and for each query $i \in [m]$, read $v$ entirely and compute the answer in query time $t=n$; or
\item use $s=m$ memory cells, where the $i$'th cell contains the answer to the query $i \in [m]$, thus allowing for query time $t=1$.
\end{inparaenum}
In a typical application, the number of queries $m=\poly(n)\gg n$, and a data structure is efficient if it uses space $s=\widetilde{O}(n)$ (or $s\ll m$) and has query time $t=\poly(\log(n))$ (or $t=n^\eps$ for a small constant $\eps>0$). Note that the two trivial solutions do not lead to such efficient data structures for $m\gg n$.

We consider randomized data structures, where both the preprocessing stage and the query stage use randomness, and are expected to output the correct answer with high probability (over the randomness of both stages). In average-case randomized data structures, the success rate of the algorithm is taken over both the inner randomness and the random input, whereas in worst-case randomized data structure, the success rate is taken only over the inner randomness of the algorithm (i.e., the algorithm succeeds with high probability \emph{on all} inputs).

We present a worst-case to average-case reduction showing that if there exists a data structure~$\DS$ that uses $s$ memory cells, has query time $t$, and success rate such that for a small fraction of inputs the data structure answers all queries correctly, then there exists another data structure $\DS'$ that uses $4s$ memory cells, has query time $4t$, and success rate such that \emph{for all inputs} the data structure answers all queries correctly with high probability.

\begin{restatable}{mtheorem}{DSstrong}\label{thm:ds-lin-strong-avg-case}
Let $\F = \F_p$ be a prime field, $\alpha\coloneqq\alpha(n) \in (0,1]$, $n,m \in \N$, and a matrix $A \in \F^{m \times n}$. Denote by $L_A$ the linear problem of outputting $\ip{A_i,x}$ on input $x \in \F^n$ and query $i\in[m]$.
    Suppose that
    \begin{equation*}
        L_A \in \DSParams{p}{s}{t}{\Pr_{x \in \F^n}[\DS_x(i) = \ip{A_i,x} \forall i \in [m]] \geq \alpha} \,.
    \end{equation*}
    Then for every $\delta>0$,
    \begin{equation*}
        L_A \in
        \DSParams{p + \exp(\log^4(1/\alpha)) \cdot \poly\log(1/\delta) \cdot \poly(n)}
        {4s + O(\log^4(1/\alpha) \log(n))}
        {4t+O(\log^4(1/\alpha) \log(n))}
        {\forall x \in \F^n \enspace \Pr[\DS'_x(i) = \ip{A_i,x} \forall i \in [m]] \geq 1-\delta}\,.
    \end{equation*}
\end{restatable}
We stress that in the average-case data structure we start with, the probability is taken over a random input (as well as the inner randomness of the algorithm), whereas in the worst-case data structure that we obtain, with high probably the algorithm is successful \emph{on all} inputs.

The reduction above shows that for any linear problem $L_A$, if a data structure succeeds on an arbitrary small constant $\alpha>0$ fraction of the inputs, then we can obtain a data structure that succeeds on all inputs with parameters that essentially differ only by a constant multiplicative factor, and the query complexity $t$ translates into query complexity $4t + O(\log(n))$.

We note that the $O(\log^4(1/\alpha) \log(n))$ overhead in the space complexity of the constructed data structure is caused by storing $O(\log^4(1/\alpha))$ numbers from $[n]$. In particular, if the word size of the data structure is $w\geq\log(n)$, then the space complexity of the resulting data structure is $4s+O(\log^4(1/\alpha))$. Similarly, in this case the query complexity of the resulting data structure is $4t+O(\log^4(1/\alpha))$.

Note that for any non-trivial data structure problem, a data structure must use at least $\Omega(n)$ memory cells (only to store a representation of the input). Therefore, even for $\alpha$ as small as $\alpha=2^{-n^\eta}$ for a small constant $\eta>0$, the overhead in the space complexity is negligible. For typical query times of data structures, such as $t=\poly(\log(n))$ and $t=n^\eps$, the overhead in the query time is negligible even for $\alpha=1/\poly(n)$ and $\alpha=2^{-n^\eta}$, respectively.

\subsubsection{Online matrix-vector multiplication}
Next we turn to the core data structure problem in fine-grained complexity, the online matrix-vector multiplication problem (OMV). In the data structure variant of this streaming problem, one needs to preprocess a matrix $M\in\F^{n\times n}$, such that given a query vector $v\in\F^n$, one can quickly compute $Mv$. The study of OMV (over the Boolean semiring) and its applications to fine-grained complexity originates from~\cite{HKNS15}, and \cite{LW17,CKL18} give surprising upper bounds for the problem. Over finite fields, \cite{FHM01,CGL15} give lower bounds for OMV, and~\cite{CKLM18} proves lower bounds for a related vector-matrix-vector multiplication problem. We prove an efficient worst-case to average-case reduction for OMV over prime fields. A concurrent and independent work~\cite{HLS21} studies worst-case to average-case reductions for OMV over the Boolean semiring and their applications.

Note that OMV is, in fact, \emph{not} a linear problem, because for a query $v$ the output is not a single field element, but rather a vector $Mv \in \F^n$. Moreover, the average case condition only guarantees success with probability taken over \emph{both} the matrix $M$ as well as the vector $v$. Nevertheless, we can exploit the fact that each coordinate of the correct output is a linear function in the entries of $M$, and extend our techniques to the more involved setting of OMV.

\begin{restatable}{mtheorem}{MV}\label{thm:MV-weak-avg-case}
    Let $\F = \F_p$ be a prime field, $n \in \N$, and $\alpha\coloneqq\alpha(n) \in (0,1]$.
    Consider the matrix-vector multiplication problem $OMV_\F$ for dimension $n$,
    and suppose that for some $\alpha>0$ it holds that
    \begin{equation*}
        OMV_\F \in \DSParams{p}{s}{t}{\Pr_{M,v}[\DS_M(v) = Mv] \geq \alpha}
        \, .
    \end{equation*}
    Then for every $\delta>0$,
    \begin{equation*}
        OMV_\F \in
        \DSParams{4p + \exp(\log^4(1/\alpha)) \cdot \poly\log(1/\delta)\cdot \poly(n)}
        {4s + O(\log^4(1/\alpha) n) + O(n^2)}
        {(4t + n) \cdot \poly(1/\alpha) \cdot \poly\log(1/\delta)}
        {\forall M,v : \Pr[\DS_M(v) = Mv] \geq 1-\delta}\,. 
    \end{equation*}
\end{restatable}

We stress that in the assumed data structure, the success rate asserts that for \emph{a random input $M$ and query~$v$}, the data structure produces the correct answer with (an arbitrary small) probability $\alpha>0$, where the probability is over
\begin{inparaenum}[(i)]
\item the random input $M$
\item random query $v$
\item and the randomness of the preprocessing and the query phases of the data structure.
\end{inparaenum}
On the other hand, the conclusion holds for worst case inputs and queries. That is, \emph{for every} input $M$ and query $v$, the obtained data structure produces the correct answer with high probability, where the probability is only over the randomness used in the preprocessing stage and the query phase of the data structure (i.e., with high probability we can compute \emph{all} of the inputs).

To understand the parameters of the reduction, note that in the the OMV problem with $n \times n$ matrices, the preprocessing must be at least $n^2$, as this is the size of the input matrix, and the query time must be at least $n$, as information-theoretically we need to output $n$ field elements. Our worst-case to average-case reduction is essentially optimal in these parameters for a constant $\alpha$, as a weak data structure that uses $s$ memory cells and query time $t$ is translated into a data structure that works for all inputs and all queries using space $4s + O(n^2)$ and query time $4t + O(n) = O(t)$.
In fact, even for $\alpha$ as small as $\alpha = 1/n^{o(1)}$, the space complexity is increased by at most $O(n^2)$, and the query time is multiplied by at most $n^{o(1)}$.

\subsubsection{Worst-case to weak-average-case reductions}
\label{sec:intro:rm}
In the following, we discuss how to obtain worst-case algorithms starting from a very weak, but natural, notion of average-case reductions that we discuss next.

Recall that in the standard definition of average-case data structures, the algorithm preprocesses its input and is then required to correctly answer all queries for an $\alpha$-fraction of all possible inputs. However, in many cases (such as in the online matrix-vector multiplication problem), we only have an average-case guarantee on both inputs and queries. In this setting, we should first ask what is a natural notion of an average-case condition.

A strong requirement for an average-case algorithm in this case is to correctly answer \emph{all queries} for at least $\alpha$-fraction of the inputs. However, it is desirable to only require the algorithm to correctly answer on an \emph{average input and query}. That is, a \emph{weak average-case} data structure for computing a function $f\colon \F^n \times [m] \to \F$ with success rate $\alpha>0$ receives an input $x\in\F^n$, which is preprocessed into $s$ memory cells. Then, given a query $i\in[m]$, the data structure $\DS_x(i)$ outputs $y\in\F^{n'}$ such that $\Pr_{x\in\F^n, i \in [m] }[\DS_x(i) = f(x,i)] \geq \alpha$.

The challenge in this setting is that the errors may be distributed between both the inputs and the queries. On one extreme, the error is concentrated on selected inputs, and then the data structure computes \emph{all queries} correctly for $\alpha$-fraction of the inputs. On the other extreme, the error is spread over all inputs, and then the data structure may only answer $\alpha$-{fraction of the queries} on any inputs. Of course, the error could be distributed anywhere in between these extremes.

While we showed that every linear problem has an efficient worst-case to average-case reduction, in \cref{sec:impossibility} we show that not all linear (and non-linear) problems admit a worst-case to \emph{weak}-average-case reductions. Nevertheless, we overcome this limitation for certain problems of interest. 

\medskip

One of the most-studied problems in static data structures is the polynomial evaluation problem~\cite{KU08,L12,DKKS21}. Here, one needs to preprocess a degree-$d$ polynomial $q \colon \F^m \to \F$ into $s$ memory cells, and then for a query $x\in\F^m$, quickly compute $q(x)$. We study the problem of evaluating a low degree polynomial in the regime where the average-case data structure might only succeed on a small $\alpha$ fraction of the queries (outside of the unique decoding regime, see discussion below). We show that we can use such an average-case data structure to obtain a worst-case data structure that can compute $q$ on any $x\in\F^m$.

\begin{restatable}{mtheorem}{evalpoly}\label{thm:DS-RM-weak-avg-case}
    Let $\F = \F_p$ be a prime field, $\alpha\coloneqq\alpha(n) \in (0,1]$, and let $m,d \in \N$ be parameters.
    Consider the problem $\RM_{\F,m,d}$ of evaluating polynomials of the form
    $q \colon \F^m \to \F$ of total degree $d$ (i.e., the problem of evaluating the Reed-Muller encoding of block length $n = \binom{m+d}{d}$).

    Suppose that
    \begin{equation*}
        \RM_{\F,m,d} \in \DSParams{p}{s}{t}{\Pr_{q,x}[\DS_q(x)] \geq \alpha}
        \enspace.
    \end{equation*}
    Then
    \begin{equation*}
        \RM_{\F,m,d} \in \DSParams{p + \exp(\log^4(1/\alpha)) \cdot \poly(n)}{4s+O(\log^4(1/\alpha) \log(n))}{O(\abs{\F}^2 \cdot t + \abs{\F} \log^4(1/\alpha) + \abs{\F} \log(n))}
        {\forall q, x : \Pr[\DS_q(x) = q(x) ] > 1 - O\left(\sqrt{\frac{d}{\abs{\F}}}\right)}
        \enspace.
    \end{equation*}
\end{restatable}

Here, similarly to \cref{thm:MV-weak-avg-case}, the assumed data structure
succeeds only for a small fraction of inputs and queries, while in the conclusion
the data structure succeeds with high probability on \emph{every input} and \emph{every query}.

As for the effect of the reduction on the parameters, we see that for any $\alpha > 1/\poly(n)$ the preprocessing time changes only by an additive $\poly(n)$, the space complexity changes from $s$ to $4s + \poly( \log(n))$, and the query time changes from $t$ to $O(\abs{\F}^2 \cdot t + \abs{\F} \cdot \poly\log(n))$. In the data structure setting, the number of queries is usually polynomial in input length. Thus, in a typical setting of parameters for $\RM_{\F,m,d}$, the field size is $|\F|=\poly(\log{n})$, and, therefore, the blow-up of $|\F|^2$ is not critical.

\paragraph{A coding-theoretic perspective.}
For small values of average-case rate $\alpha>0$, the polynomial evaluation problem can be cast as \emph{list decoding with preprocessing}, by viewing the outputs of the query phase of the data structure as a function $h \colon \F^m \to \F$ that agrees with the input polynomial $q\colon \F^m \to \F$ on some small fraction of the queries, and the goal is to recover $q$ from $h$.

Indeed, note that for a small $\alpha>0$, if a function $h \colon \F^m \to \F$ agrees with some unknown low-degree polynomial $q$ on $\alpha$ fraction of the inputs, then there are potentially $O(1/\alpha)$ possible low-degree polynomials that are equally close to $h$. Hence, without preprocessing it is impossible to recover the original polynomial $q$. However, in the data structure settings, we can use the preprocessing to obtain an auxiliary structural information that would later allow us to transition from the list decoding regime to the unique decoding regime, and in turn, compute the values of the correct polynomial~$q$ with high probability (see more details in \cref{tech:beyond}).

\subsection{Open problems}
Our work leaves many natural open problems, such as obtaining reductions for various natural problems in other computational models (e.g., communication complexity, property testing, PAC learning, and beyond). However, for brevity, we would like to focus on and highlight one direction that we find particularly promising.

In \cref{thm:ds-lin-strong-avg-case}, we design worst-case to average-case reductions for linear problems in the setting of static data structures. An immediate and alluring question is whether our local correction via additive combinatorics framework can also be used to show worst-case to average-case reductions for all linear problems for both circuits and uniform algorithms. We observe that using similar techniques as in \cref{thm:ds-lin-strong-avg-case}, our framework can be used to show that given an efficient average-case circuit or uniform algorithm and an efficient \emph{verifier} for the problem, one can indeed design an explicit efficient worst-case circuit or uniform algorithm. A natural open problem here is to eliminate the assumption about the verifier and answer the aforementioned question to the affirmative.

\subsection*{Acknowledgments}
We are grateful to Tom Sanders for providing a sketch of the proof of the probabilistic version of the quasi-polynomial Bogolyubov-Ruzsa lemma. We would also like to thank Shachar Lovett and Tom Sanders for discussions regarding the quasi-polynomial Bogolyubov-Ruzsa lemma. 

\section{Technical overview}
\label{sec:overview}

We provide an overview of the main ideas and techniques that we use to obtain our results. For~concreteness, we illustrate our techniques by first focusing on the matrix multiplication problem.
    
We start in \cref{tech:challange}, where we explain the challenge and discuss why the naive approach fails. In \cref{tech:ac} we present the technical components that lie at the heart of this work: \emph{local correction lemmas via additive combinatorics}. Equipped with these technical tools, in \cref{tech:matrix} we present the main ideas in our worst-case to average-case reduction for matrix multiplication. Finally, in \cref{tech:beyond} we briefly discuss how to obtain the rest of our main results.
    
\subsection{The challenge: low-agreement regime}
\label{tech:challange}

Recall that in the matrix multiplication problem we are given two matrices $A,B \in \F^n$, and the goal is to compute their matrix product $A \cdot B$. For simplicity of the exposition, unless specified otherwise, in this overview we restrict our attention to the field $\F_2$, and to constant values of the success rate parameter $\alpha>0$ of average-case algorithms.

We would like to show that if there is an \emph{average-case} algorithm $\ALG$ that can compute matrix multiplication for an $\alpha$-fraction of all pairs of matrices $A,B \in \F^n$ in time $T(n)$, then there is a \emph{worst-case} randomized algorithm $\ALG'$ that runs in time $O(T(n))$ and computes $A \cdot B$ with high probability for \emph{every} pair of matrices $A$ and $B$.

We start with the elementary case where the average-case guarantee is in the \emph{high-agreement regime}, i.e., where the algorithm succeeds on, say, $99\%$ of the inputs; that is, 
\begin{equation}
    \label{eq:tech:high_agr}
    \Pr_{A,B \in \F^{n \times n}}[\ALG(A,B)=A\cdot B] \geq \alpha \;,
\end{equation}
for $\alpha=0.99$.
In this case, a folklore local correction procedure (see, e.g., \cite{BlumLR90}) will yield a worst-case algorithm that succeeds with high probability on all inputs. We next describe this procedure.

Given an \emph{average-case} algorithm $\ALG$ satisfying \cref{eq:tech:high_agr} with $\alpha=0.99$, consider the worst-case algorithm $\ALG'$ that receives any two matrices $A,B \in \F^{n \times n}$ and first samples uniformly at random two matrices $R,S \in \F^{n \times n}$. Next, writing $A = R + (A-R)$ and $B = S + (B-S)$, the algorithm $\ALG'$ computes
\begin{equation}
    \label{eq:tech:self-corr}
    M = \ALG(R,S) + \ALG(A-R,S) + \ALG(R,B-S) + \ALG(A-R,B-S) \;.    
\end{equation}

Denote by $X$ the set of matrix pairs $(A,B)$ for which $\ALG(A,B)=A\cdot B$, and recall that by \cref{eq:tech:high_agr} the density of $X$ is $0.99$. Note that:
(a) the matrices $R$, $A-R$, $S$, and $B-S$ are uniformly distributed, and
(b) if the pairs $(R,S)$, $(A-R,S)$, $(R,B-S)$, and $(A-R,B-S)$ are in the set $X$, then by \cref{eq:tech:self-corr} we have $M = A\cdot B$, and the algorithm $\ALG'$ computes the multiplication correctly. Hence, by a union bound we have $\Pr[M = AB] \geq 1 - 4 \cdot 0.01 > 0.9$ \emph{for all} matrices $A,B \in \F^{n \times n}$. Of course, the error probability can be further reduced by repeating the procedure and ruling by majority.

Unfortunately, this argument breaks when the average-case guarantee is weaker; namely, in the \emph{low-agreement regime}, where the algorithm succeeds on, say, only $1\%$ of the inputs. Here, when trying to self-correct as above, the vast majority of random choices would lead to a wrong output, and so at a first glace, the self-correction approach may seem completely hopeless.\footnote{Indeed, consider the counterexample where the average-case algorithm $\ALG(A,B)$ outputs $A \cdot B$ in case the first element of $A$ is $0$ and returns the zero matrix in case the first element of $A$ is $1$. Note that in this case $\Pr_{A,B \in \F^{n \times n}}[\ALG(A,B)=A\cdot B] \geq 1/2$, yet no decomposition of $A=\sum_i A_i$ and $B=\sum_i B_i$ as described above could self-correct matrix multiplication where the first element of $A$ is $1$. Indeed, any such composition would have an $A_i$ with the first element $1$, where $\ALG(A_i,B_j)$ fails.}

Nevertheless, using more involved tools from additive combinatorics such as a probabilistic version of the quasi-polynomial Bogolyubov-Ruzsa lemma that we show, as well as tools such as small-biased sample spaces and the Goldreich-Levin algorithm, we can construct different local correction procedures that work in the \emph{low-agreement regime}. We proceed to describe our framework for local correction using the aforementioned tools.

\subsection{Local correction via additive combinatorics}
\label{tech:ac}

Additive combinatorics studies approximate notions of algebraic structures via the perspective of combinatorics, number theory, harmonic analysis and ergodic theory. Most importantly for us, it provides tools for transitioning between algebraic and combinatorial notions of approximate subgroups with only a small loss in the underlying parameters (see surveys \cite{Lovett15,Lovett17}).

The starting point of our approach for local correction is a fundamental result in additive combinatorics, known as \emph{Bogolyubov's lemma}, which shows that the $4$-ary sumset of any dense set in $\F_2^n$ contains a large linear subspace. More accurately, recall that the sumset of a set $X$ is defined as $X+X = \{ x_1+x_2 \,:\, x_1,x_2 \in X \}$, and similarly $4X = \{ x_1+x_2+x_3+x_4\,:\, x_1,x_2,x_3,x_4 \in X \}$. These quantities can be thought of as quantifying a combinatorial analogue of an approximate subgroup. Bogolyubov's lemma states that for any subset $X \subseteq \F_2^n$ of density $|X|/2^n \geq \alpha$, there exists a subspace $V \subseteq 4X$ of dimension at least $n - \alpha^{-2}$.

We will show that statements of the above form can be used towards obtaining a far stronger local correction paradigm than the one outlined in \cref{tech:challange}. To see the initial intuition, consider an average-case algorithm that is guaranteed to correctly compute $\alpha$-fraction of the inputs, and denote by $X$ the set of these correctly computed inputs. Then $|X|/2^n \geq \alpha$, and Bogolyubov's lemma shows that there exists a large subspace $V$ such that every $v \in V$ can be expressed as a sum of four elements in $X$, each of which can be computed correctly by the average-case algorithm.

The approach above suggests a paradigm for local correction, however, there are several non-trivial problems in implementing this idea. For starters, how could we handle inputs that lay \emph{outside of the subspace} $V$? To name a few others: how can we amplify the success probability in the low-agreement regime? How do we algorithmically obtain the decomposition?  Can we handle finite fields beyond $\F_2^n$? How do we handle average-case where the success rate $\alpha$ is sub-constant? 

Indeed, for our worst-case to average-case reductions, we will need local correction lemmas with stronger structural properties than those admitted by Bogolyubov's lemma, as well as new ideas for each one of the settings. In the following, we discuss the main hurdles for the foregoing approach and the tools that are needed to overcome them, leading to our main technical tool, which is a probabilistic version of the quasi-polynomial Bogolyubov-Ruzsa lemma that we obtain. Then, we present our framework for local correction using these techniques. Finally, in \cref{tech:matrix,tech:beyond} we show the additional ideas that are necessary for applying the local correction lemmas in the settings of matrix multiplication, online matrix-vector multiplication, and data structures.

\paragraph{A probabilistic Bogolyubov lemma.} An immediate problem with the aforementioned local correction scheme is that while Bogolyubov's lemma asserts that \emph{there exists} a decomposition of each input into a sum of four elements in $X$, it does not tell us how to obtain this decomposition.

Toward this end, we further show that each vector $v \in V$ has many ``representations'' as a sum of four elements from $X$. This way, for any $v \in V$ we can efficiently sample a representation $v = x_1 + x_2 + x_3 + x_4$, where each $x_i \in X$. More accurately, let $X \seq \F_2^n$ be a set of density $\alpha$, let $R = \{r \in \F^n \setminus \{0\} : \abs{\hat{1}_X(r)} \geq \alpha^{3/2}\}$, and let $V = \{v \in \F^n : \ip{v,r} = 0 \ \forall r \in R\}$ be a linear subspace defined by $R$. Then $|R| \leq 1/\alpha^2$ and for all $v \in V$ it holds that
    \begin{equation*}
        \Pr_{x_1,x_2,x_3}[x_1,x_2,x_3,v - x_1 - x_2 - x_3 \in X] \geq \alpha^5 \;.
    \end{equation*}

\paragraph{Sparse-shift subspace decomposition.} The probabilistic Bogolyubov lemma allows us to locally correct inputs inside the subspace $V \subseteq 4X$. However, we need to be able to handle any vector in the field. Towards that end, we show an algebraic lemma that allows us to decompose each element of the field into a sum of an element $v$ in the subspace $V$ and a \emph{sparse} shift-vector~$s$. More accurately, let $R \subseteq \F^n \setminus \{\vec{0}\}$ and $V = \{v \in \F^n : \ip{v,r} = 0 \; \forall r \in R\}$. We show that there exists a collection of $t\leq \abs{R}$ vectors $B=\{b_1,\dots,b_t\},\,b_i \in \F^n$ and indices $k_1,\dots, k_t \in [n]$ such that $\sp(B) = \sp(R)$ and every vector $y \in \F^n$ can be written as $y = v + s$, where $v \in V$ and $s = \sum_{j=1}^t c_j \cdot \vec{e}_{k_j}$ for $c_j = \ip{y, b_j}$ and $\vec{e}_{k_j}$ is a unit vector.
    
We stress that the sparsity of the decomposition is essential to our applications, as we cannot locally correct the shift part of the decomposition, and instead we need to compute it explicitly. We remark that for matrix multiplication we can obtain a stronger guarantee by dealing with matrices outside of the subspace $V$ via a low-rank random matrix shifts (see \cref{tech:matrix}).

\paragraph{Subspace computation via the Goldreich-Levin lemma.} In order to perform local correction using additive combinatorics machinery as above while maintaining computational efficiency, we need to be able to compute the aforementioned basis $b_1,\dots,b_t \in \F^n$ and indices $k_1,\dots, k_t \in [n]$ efficiently. We note that, in essence, this problem reduces to learning the heavy Fourier coefficients of the set $X$. Thus, using ideas from \cite{BenSassonRZTW14} and an extension of the Goldreich-Levin algorithm to arbitrary finite fields, we can perform the latter in a computationally efficient way.

\paragraph{Probabilistic quasi-polynomial Bogolyubov-Ruzsa lemma.}
The main weakness of Bogolyubov's lemma is that the co-dimension of the subspace that it admits is polynomial in $1/\alpha$, where $\alpha$ is the success rate of the average-case algorithm. While this dependency on $\alpha$ allows us to locally correct in the $1\%$ agreement regime, it becomes degenerate when $\alpha$ tends to $0$ rapidly.

A natural first step towards overcoming this barrier is to use a seminal result due to Sanders~\cite{Sanders12}, known as the \emph{quasi-polynomial Bogolyubov-Ruzsa lemma}, which shows the existence of a subspace whose co-dimension's dependency on $1/\alpha$ is \emph{exponentially} better. That is, the lemma shows that for a set $X \seq \F_2^n$ of size $\alpha \cdot \abs{\F_2}^n$, where $\alpha \in (0,1]$, there exists a subspace $V \seq \F_2^n$ of dimension $\dim(V) \geq n - O(\log^4(1/\alpha))$ such that $V \seq 4X$. However, as in the case of Bogolyubov's lemma, we have the problem that the statement is only \emph{existential}.

We thus prove a probabilistic version of the quasi-polynomial Bogolyubov-Ruzsa lemma (see \cref{thm:bogolyubov-quasipoly}) over any field $\F = \F_p$, which asserts that for an $\alpha$-dense set $X \seq \F^n$, there exists a subspace $V \seq \F^n$ of dimension $\dim(V) \geq n - O(\log^4(1/\alpha))$ such that for all $v \in V$ it
holds that
    \begin{equation*}
        \Pr_{x_1,x_2,x_3 \in \F^n}[x_1, x_2 \in A, x_3, x_4 \in -A] \geq \Omega(\alpha^5)
        \enspace,
    \end{equation*}
where $x_4 = v-x_1-x_2-x_3$.
Furthermore, by combining the techniques above, we show that given a query access to the set $X$, there is an algorithm that runs in time $\exp(\log^4(1/\alpha)) \cdot \poly\log(1/\delta) \cdot \poly(n)$ and with probability $1-\delta$ computes a set of vectors $R \seq \F^n$ such that $V = \{v \in \F^n : \ip{v,r} = 0 \; \forall r \in R\}$.

We are grateful to Tom Sanders for providing us with the argument for showing this lemma, and we provide the proof in \cref{sec:appendix-sanders}.

\paragraph{Our local correction lemma.}
We are now ready to provide an informal statement of our local correction lemma, which builds on the machinery above, and in particular, on the probabilistic quasi-polynomial Bogolyubov-Ruzsa lemma. 

Loosely speaking, our local correction allows us to decompose any vector $y\in\F^n$ as a linear combination of the form
\begin{equation*}
    y = x_1 + x_2 - (x_3 + x_4) + s
    \;,
\end{equation*}
where $x_1,x_2,x_3,x_4 \in X$ and $s \in \F^n$ is a \emph{sparse} vector.

\begin{lemma}[informally stated, see \cref{cor:y=4x+s}]
    For a field $\F = \F_p$ and $\alpha$-dense set $X \seq \F^n$, there exists $t \leq 1/\alpha^2$ vectors $b_1,\dots,b_t \in \F_2^n$ and indices $k_1,\dots, k_t \in [n]$ satisfying the following. Given a vector $y \in \F_2^n$, let $s = \sum_{j=1}^t \ip{y, b_j} \cdot \vec{e}_{k_j}$ we have
    \begin{equation*}
        \Pr_{x_1,x_2,x_3 \in \F^n}[ x_1, x_2 \in X, x_3, x_4  \in -X] \geq \Omega(\alpha^5)
        \enspace,
    \end{equation*}
    where $x_4=y - s - x_1 - x_2 - x_3$.

    Furthermore, given an oracle that computes $1_X(x)$ with probability at least $2/3$, there exists an algorithm that makes $\exp(\log^4(1/\alpha)) \cdot \poly\log(1/\delta) \cdot \poly(n)$ oracle calls and field operations, and with probability at least $1-\delta$ outputs $b_1,\dots,b_t$ and $k_1,\dots, k_t$.
\end{lemma}

The aforementioned local correction lemmas lie at the heart of our average-case to worst-case reductions, which we discuss next.

\subsection{Illustrating example: matrix multiplication}
\label{tech:matrix}
    
    We present a high-level overview of our reductions for matrix multiplication, which illustrates the key ideas that go into the proof. Let $\ALG$ be an average-case algorithm that can compute matrix multiplication for an $\alpha$-fraction of all pairs of matrices $A,B \in \F^n$ in time $T(n)$.
    We use the \emph{average-case} algorithm $\ALG$ to construct a \emph{worst-case} randomized algorithm $\ALG'$ that runs in time $O(T(n))$ and computes $A \cdot B$ with high probability for \emph{every} pair of matrices $A$ and $B$.
    For simplicity of the exposition, in this overview we make the following assumptions:
    (1) the algorithm $\ALG$ is \emph{deterministic},
    (2) the input is a pair $(A,B)$ such that $A$ is a matrix satisfying $\Pr_{B'}[\ALG(A,B') = A \cdot B'] \geq \alpha$,
    (3) the success rate $\alpha$ is a constant, and
    (4) the field $\F$ is $\F_2$.
    
    We start by noting two simple facts. First, given the algorithm's (potentially wrong) output $\ALG(A,B)$, we can efficiently check whether the computation is correct using Freivalds' algorithm (\cref{lemma:freivalds}). Second, denoting by $X = \{B' \in \F_2^{n \times n} : \ALG(A,B') = A \cdot B'\}$ the set of ``good'' matrices, we have that if $B \in X$, then the average-case algorithm correctly outputs $\ALG(A,B) = A \cdot B$. Hence, the main challenge is in dealing with the case that $B \notin X$, in which we need to locally correct the value of the multiplication.
    
    \paragraph{Local correction via Bogolyubov's lemma.}
    The first idea is to reduce the problem to the case where the set of good matrices contains a large subspace, and hence admits local correction, as discussed in \cref{tech:ac}. Specifically, by the probabilistic Bogolyubov lemma, given $X$ we can choose a subspace $V \seq \F_2^{n \times n}$ of matrices, where
    $\dim(V) \geq n^2 - 1/\alpha^2$, such that for any $B' \in V$, if we sample $M_1,M_2,M_3$ uniformly at random, then 
    \begin{equation*}
        \Pr[M_1,M_2,M_3,M_4 \in X] \geq \alpha^5 \;,\; \text{where } M_4 = B' - M_1 - M_2 - M_3 \;.
    \end{equation*}
    Note that if the matrices $M_1,M_2,M_3,M_4$ produced by our sampling are all in the set of good matrices $X$, then we can self-correct the value of $\ALG(A,B')$ by evaluating $\{\ALG(A,M_i)\}_{i \in [4]}$ and computing the linear combination
    \begin{equation*}
        \sum_{i=1}^4 \ALG(A,M_i) = \sum_{i=1}^4 A \cdot M_i = A \cdot (\sum_{i=1}^4 M_i) = A \cdot B' \;.
    \end{equation*}
    
    Note that this event is only guaranteed to occur with probability $\alpha^5$, which is far smaller than $1/2$. Nevertheless, since we can verify the computation using Freivalds' algorithm, we can boost this probability to be  arbitrarily close to 1 by repeating the random sampling step $O(1/\alpha^5)$ times, each time computing $\sum_{i=1}^4 \ALG(A,M_i)$ and verifying if the obtained result is indeed correct using Freivalds' algorithm. Therefore, if $B$ belongs to the (unknown) subspace $V$, then the algorithm described above indeed computes $A \cdot B$ with high probability
    in time $O(T(n)/\poly(\alpha)) = O(T(n))$.
    
    However, the approach above does not work for matrices $B$ that do not lie in the subspace $V$ described above. To deal with this case, our next goal is to ``shift'' the matrix into the subspace $V$ using low-rank random shifts, which can then be computed efficiently and used for local correction. We describe this procedure next.
    
    \paragraph{Low-rank random matrix shifts.} 
    We start by making the following key observation: if we have an arbitrary matrix $A$, and a matrix $B \in \F^{n \times n}$ of rank $k$, then their product $AB$ can be computed in time $O(kn^2)$, given a rank-$k$ decomposition of $B$. Details follow.
    
    To see this, suppose that the first $k$ columns of $B$ denoted by $(B_i)_{i=1}^k$, are linearly independent, and for each of the remaining $n-k$ columns $(B_j)_{j=k+1}^n$, we know the linear combination $B_j = \sum_{i=1}^k d_{i,j} \cdot B_i$ for some coefficients $d_{i,j} \in \F$. We can first multiply $A$ by each of the $k$ linearly independent columns of $B$. Then, to compute the remaining columns, for each $i=1, \dots, k$ let $C_i = A \cdot B_i$ be the $i$'th column of the matrix $C = AB$, and observe that if $B_j = \sum_{i=1}^k d_{i,j} B_i$, then $C_j = A \cdot B_j = A \cdot (\sum_{i=1}^k d_{i,j} B_i) = \sum_{i=1}^k d_{i,j} \cdot C_i$, which can be computed in $O(kn)$ time for each $j$. Therefore the total running time of multiplying $A$ by $B$ is $O(kn^2)$.
    
    We are now ready to describe our method for shifting the matrices into the subspace $V$ using low-rank matrices, capitalizing on the observation above. Given the matrix $B$ (that is, possibly, not in $V$), we sample a random matrix $R_B \in \F^{n \times n}$ of rank $2k$ by randomly choosing $2k$ columns and filling them with uniformly random field elements.
    Note that with high probability these $2k$ columns are linearly independent. Then, we let the rest of the columns be random linear combinations of the first $2k$ columns we chose. We observe that if $\dim(V) = n-k$, then 
    \begin{equation*}
        \Pr[B+R_B \in V] \geq \frac{1}{2|\F|^{k}} \;.
    \end{equation*}
    If indeed $B+R_B \in V$, then we can compute $A \cdot(B+R_B)$ using the procedure discussed above, by writing $B+R_B$ as a sum of 4 random matrices $B+R_B = M_1+M_2+M_3+M_4$, applying $\ALG(A,M_i)$ for each $i=1...4$,
    and using Freivalds' algorithm to efficiently check if the produced output is correct or not.
    
    Note that since we have a lower bound on the probability that $B+R_B$ belongs to the desired subspace, we have an upper bound on the expected number of attempts required until this event occurs. When we obtain such low-rank matrix shifts, which we verify using Freivalds' algorithm, we proceed by computing $A \cdot R_B$.
    Since $R_B$ is a matrix of rank at most $2k$, the total running time of this will be $O(kn^2)$. Finally, we return 
    \begin{equation*}
        \ALG(A, B+R_B)- A \cdot R_B \;,
    \end{equation*}
    which indeed produces the correct answer assuming that $\ALG(A, B+R_B)$ is correct.
    
    \begin{remark}
        The discussion above made the simplifying assumption that the inputs we are getting are pairs $(A,B)$ such that $A$ is a matrix satisfying $\Pr_{B'}[\ALG(A,B') = A \cdot B'] \geq \alpha$.
        The actual proof require also handling the inputs for which the matrix $A$ does not satisfy this requirement, which is done using similar ideas by applying the local correction procedure first to $A$ and then to $B$.
    \end{remark}
\subsection{Beyond matrix multiplication}
\label{tech:beyond}

We conclude the technical overview by briefly sketching some of the key ideas in the rest of our worst-case to average-case reductions, building on the local correction lemmas outlined in \cref{tech:ac}. Below we assume that all data structures are deterministic, but by standard techniques this assumption is without loss of generality. We start with the simplest setting, and then proceed to the more involved ones.

\paragraph{Worst-case to average-case reductions for all linear data structure problems.}
The setting here is the closest to that of matrix multiplication. Let $\DS_A$ be an average-case data structure for a linear problem defined by $A$, where we preprocess an input vector $x$ and the answer to query $i$ is $\ip{A_i,x}$, and $A_i$ is the i'th row of $A$.

Given a vector $y \in \F^n$, we use our local correction lemma to obtain a decomposition of the form $y = x_1 + x_2 - ( x_3 + x_4) + v$, where $x_1,x_2,x_3,x_4 \in X$ (i.e., on which $DS_{x_j}(i) = \ip{A_i,x_j}$ for all $i$) and a sparse shift vector $v = \sum_{j=1}^t \ip{y, b_j} \cdot \vec{e}_{k_j}$. We then preprocess each of the $x_j$'s by applying $DS_{A}$ to it, and we also compute $\ip{A_i,v}$ efficiently by using its sparse representation. The idea is that by the linearity of the problem, we can locally correct according to $\sum_{j=1}^4 DS_{x_j}(i) + \ip{A_i,v}$.

It important to note that, unlike in the setting of matrix multiplication, we cannot use the random low-rank matrix shifts, nor Freivald's algorithm for verification. However, this is where we rely on the sparse subspace decomposition to shift the input into the subspace $V$ implied by the quasi-polynomial Bogolyubov-Ruzsa lemma. In addition, instead of relying on Freivalds' algorithm for verification, here we use the guarantee about the correctness of computation in the subspace $V$ together with the sparsity of the shift vector, which allows us to correct its corresponding contribution via explicit computation. See details in \cref{sec:all_DS}.

\paragraph{Online matrix-vector multiplication (OMV).}
The online setting of the OMV problem poses several additional challenges. Recall that in the average-case reductions above,  the input is a vector $x \in \F_2^n$, each query is a coordinate $i \in [n]$, and the matrix $M \in \F_2^{n \times n}$ is a hard-coded parameter. In the OMV problem, the matrix $M$ is the \emph{input}, the vector $x$ is the \emph{query}, and answer to a query is not a scalar but rather a \emph{vector}. Hence we need to use a two-step local correction where we first decompose the matrix and then decompose the vector. Observe that we can use our additive combinatorics mechanism to preprocess the matrix $M$ and get a description of the subspace $V$ that it asserts, as well as the formula that is required to compute the shift vector $s$ given $x$, but the problem is that here we cannot preprocess $x$, as it arrives online. Thus, in the query phase, when the algorithm receives $x$, we want to find the decomposition $x=x_1+x_2+x_3+x_4+s$. We then compute the shift vector $s$, and then sample $x_i$'s whose sum is $x-s$. However this leaves us with the task of checking that all of the $x_i$'s are computed correctly. To this end, we rely on a generalization of \emph{small-bias sample spaces} to finite fields in order to obtain an efficient verification procedure. See details in \cref{sec:OMV}.

\paragraph{Weak-average-case reductions.}
As discussed in the introduction, in the setting of weak-average-case we cannot expect a reduction for all linear problems. In turn, this leads to substantially different techniques. We concentrate on the multivariate polynomial evaluation problem. Here, we are given a polynomial $p \colon \F^m \to \F$ of degree $d$, where for simplicity, in this overview we fix the parameters $d = \log(n)$, $|\F| = \poly(\log(n))$, and $m = \log(n)/\log\log(n)$, so that we encode $n$ field elements using a codeword of length $\poly(n)$, and the distance is $1 - dm/|\F| > 0.99$. The polynomial is given as input by its $n = d^m$ coefficients, the queries are of the form $x \in \F^m$, and the goal is to output $p(x)$. The key difficulty here, is that for small values of the average-case rate $\alpha>0$, we need to be able to deal with the \emph{list decoding regime} (see discussion in \cref{sec:intro:rm}).

The first step is to rely on our additive combinatorics local correction tools similarly as in the OMV reduction. Here the idea is to preprocess the polynomial $p$ and obtain a decomposition of the form $p = p_1 + p_2+p_3+p_4+s$, where again $s$ is a sparse shift-vector. We then construct a data structure for each $p_i$. However, since we cannot process the queries $x \in \F^m$, we are left with the task of locally correcting the noisy polynomials $\{p_i\}$. If a polynomial $p_i$ is only slightly corrupted (i.e., within the unique decoding regime), we can easily locally correct it without using any preprocessing. However, we also need to deal with noisy polynomials $p_i$ in the \emph{list decoding regime} in which only $\alpha$-fraction of the points are evaluated correctly, for an arbitrarily small $\alpha$.

We overcome the difficulty above by capitalizing the preprocessing power of the data structure. Namely, we will show how to boost the success probability from the list-decoding regime to the unique-decoding regime, in which case we can perfectly correct the polynomial via the local correction of the Reed--Muller code. The key idea is that by the generalized Johnson bound, there is only a list of $O(1)$ codewords that agree with the average-case data structure on at least $\alpha/2$-fraction of the points. We thus fix a reference point $w \in \F^m$ and explicitly compute the correct value of $p_i(w)$. Next, we sample a random point $r$ and query the points of line $\ell_{x,w}$ incident to $r$ and the reference point $z$. Then, we consider the list (of size $O(1)$) of all low-degree univariate polynomials that agree with the queried points on $\ell_{x,w}$, and trim the list by removing each polynomial that does not agree on the reference point. Using the sampling properties of lines in multivariate polynomials, we can show that answering accordingly to the remaining polynomials in the list would yield the right value with high probability.

\section{Additive combinatorics toolbox}
\label{sec:AC}

In this section, we provide a toolbox for locally correcting vectors using techniques from additive combinatorics. The toolkit will play a key technical role in all of our worst-case to average-case reductions.

Throughout this section we fix a finite field $\F$. For simplicity, we set $\F = \F_p$ for a prime number~$p$. However, we remark that the following results hold for any finite field, with only a negligible change in parameters (see discussions in relevant places below). Recall that the sumset of a set $X$ is defined as $X+X = \{ x_1+x_2 \,:\, x_1,x_2 \in X \}$, and, similarly, $t \cdot X = \{ x_1+ \ldots + x_t\,:\, x_1, \ldots ,x_t \in X \}$ for an integer $t\geq1$.

Let $X \subset \F^n$ be a subset of size $\abs{X} = \alpha \cdot \abs{\F}^n$. As we outlined in \cref{sec:overview}, our goal is to decompose any vector $y\in\F^n$ as a linear combination of the form
\begin{equation*}
    y = x_1 + x_2 - (x_3 + x_4) + s
    \;,
\end{equation*}
where $x_1,x_2,x_3,x_4 \in X$, and $s \in \F^n$ is a sparse vector.

Towards this end, we will need additive combinatorics lemmas that will allow us to find a large subspace $V \subseteq 2X - 2X$, so that any vector $v \in V$ can be written as $v = x_1 + x_2 - (x_3 + x_4)$. Crucially, we will show that we can efficiently sample such a decomposition and verify membership in the subspace $V$.

\subsection{Probabilistic and quasi-polynomial Bogolyubov-Ruzsa lemmas}
\label{sec:bogo}

A natural starting point for obtaining a subspace as discussed above is via \emph{Bogolyubov's lemma}, which states that for any subset $X \subseteq \F_2^n$ of density $|X|/2^n \geq \alpha$, there exists a subspace $V \subseteq 4X$ of dimension at least $n - \alpha^{-2}$. However, in addition to minor issues such as being restricted to the field $\F_2$, there are some fundamental problems with using Bogolyubov's lemma for local correction. Most importantly for our application is that while Bogolyubov's lemma asserts that \emph{there exists} a decomposition of each input into a sum of four elements in $X$, it does not tell us how to obtain this decomposition.

Hence, we further show that each vector $v \in V$ has many ``representations'' of a sum of 4 elements from $X$. This way, for any $v \in V$ we can efficiently sample a representation $v = x_1 + x_2 + x_3 + x_4$, where each $x_i \in X$. We refer to this statement as the probabilistic Bogolyubov lemma. To make the following discussion precise, we shall need the following notation. 

Given a set $X \seq \F^n$, we denote by $1_X \colon \F^n \to \{0,1\}$ the indicator function of the set $X$.
The convolution of two boolean functions $f$ and $g$ we denote by $(f * g)(x) = \E_y[f(y)g(x-y)]$.
The Fourier expansion of a function $f \colon \F^n \to \C$ is given by $f(x) = \sum_{r \in \F^n} \hat{f}(r) \cdot \chi_r(x)$,
where the Fourier coefficients of $f$ are defined as $\hat{f}(r) = \ip{f,\chi_r} = \E_x[f(x) \cdot \overline{\chi_r(x)}]$,
with $\chi_r(v) = \omega^{\ip{v,r}}$ and $\omega = e^{\frac{2 \pi i}{p}}$ is the $p$'th root of unity.
In particular for convolution of two functions we have $(f * g)(x) = \sum_r \hat{f}(r)\hat{g}(r) \chi_r(x)$.

\begin{lemma}[Probabilistic Bogolyubov lemma]
    Let $\F = \F_p$ be a prime field, and 
    let $X \seq \F^n$ be a set of size $\abs{X} = \alpha \cdot \abs{\F}^n$ for some $\alpha \in (0,1]$.
    Let $R = \{r \in \F^n \setminus \{0\} : \abs{\hat{1}_X(r)} \geq \alpha^{3/2}\}$,
    and let $V = \{v \in \F^n : \ip{v,r} = 0 \ \forall r \in R\}$ be a linear subspace defined by $R$.
    Then $|R| \leq 1/\alpha^2$, and for all $v \in V$ it holds that
    \begin{equation*}
        \Pr_{x_1,x_2,x_3}[x_1,x_2,x_3,v - x_1 - x_2 - x_3 \in X] \geq \alpha^5 \;.
    \end{equation*}
\end{lemma}

\begin{proof}
    Note first that by Parseval's identity we have $\alpha = \ip{1_X,1_X} = \norm{1_X}_2^2 = \sum_{r} \abs{\hat{1}_X(r)}^2$.
    In particular, for $R = \{r \in \F^n \setminus \{0\} : \abs{\hat{1}_X(r)}^2 > \alpha^3\}$
    we have $|R| \leq \frac{\alpha}{\alpha^3} = \frac{1}{\alpha^2}$.
    Furthermore, we have
    \begin{equation*}
      \sum_{r \in \F^n \setminus (R \cup \{0\})} \abs{\hat{1}_X(r)}^4
      \leq \alpha^3 \cdot \sum_{r \in \F^n \setminus (R \cup \{0\})} \abs{\hat{1}_X(r)}^2
      \leq \alpha^3 (\alpha - \alpha^2) \leq \alpha^4 - \alpha^5
      \enspace,
    \end{equation*}
    where the second inequality uses that $\sum_r \abs{\hat{1}_X(r)}^2=\alpha$, and $\abs{\hat{1}_X(0)}^2=\alpha^2$.
    
    Noting that for every $v \in V$ we have $\chi_r(v) = \omega^{\ip{v,r}} = \omega^0 = 1$ for all $r \in R$,
    it follows that
    \begin{eqnarray*}
        \Pr_{x_1,x_2,x_3 \in \F^n}[x_1,x_2,x_3,v-x_1-x_2-x_3 \in V]
        & = & (1_X* 1_X * 1_X * 1_X)(v) \\
        & = & \sum_{r \in \F^n} (\hat{1}_X(r))^4 \chi_r(v) \\
        & = & \abs{\hat{1}_X(0)}^4 \chi_0(v) + \sum_{r \in R} \abs{\hat{1}_X(r)}^4 \chi_r(v) \\
        && + \sum_{r \in \F^n \setminus (R \cup \{0\})} \abs{\hat{1}_X(0)}^4 \chi_r(v) \\
        & \geq & \alpha^4 + \abs{R} \cdot \alpha^6 - (\alpha^4 - \alpha^5) \\
        & \geq & \alpha^5
        \enspace,
    \end{eqnarray*}
    as required.
\end{proof}

In fact, the foregoing lemma suffices for our application for worst-case to average-case reductions where the success rate $\alpha$ is \emph{a constant}. However, to also allow for success rates that tend to zero, we shall need a much stronger statement of the form of the quasi-polynomial Bogolyubov-Ruzsa lemma, due to Sanders~\cite{Sanders12} (see also \cite{Lovett15, BenSassonRZTW14}), which admits an exponentially better dependency on $\alpha$, albeit without the efficient sampling property.

\begin{lemma}[Quasi-polynomial Bogolyubov-Ruzsa lemma~\cite{Sanders12}]
\label{thm:bogolyubov-quasipoly1}
    Let $\F = \F_p$ be a prime field,
    and let $X \seq \F^n$ be a set of size $\alpha \cdot \abs{\F}^n$ for some $\alpha \in (0,1]$.
    There exists a subspace $V \seq \F^n$ of dimension $\dim(V) \geq n - O(\log^4(1/\alpha))$
    such that $V \seq 2X - 2X$.
\end{lemma}

The caveat, however, is that while in \cref{thm:bogolyubov-quasipoly1} the codimension of $V$ is only \emph{polylogarithmic} in $1/\alpha$ (as opposed to polynomial, as in the probabilistic Bogolyubov lemma), it only guarantees that for each $v \in V$ \emph{there exist} $x_1,x_2,x_3 ,x_4 \in X$ such that $x_1+x_2+x_3+x_4=v$.

Hence, we further show that each vector $v \in V$ has many ``representations'' in $2X-2X$. In particular, for any $v \in V$ we can efficiently sample a representation $v = x_1 + x_2 - x_3 - x_4$, where each $x_i \in X$. We are grateful to Tom Sanders for providing us with a modification of his proof that admits a probabilistic version of the quasi-polynomial Bogolyubov-Ruzsa lemma.
Furthermore, we rely on the Goldreich-Levin algorithm and the techniques in \cite{BenSassonRZTW14} to obtain an efficient algorithm for verifying membership in the implied subspace. This yields the main technical tool that underlies our local correction paradigm.

\begin{restatable}[Probabilistic quasi-polynomial Bogolyubov-Ruzsa lemma]{lemma}{sandersThm}\label{thm:bogolyubov-quasipoly}
    Let $\F = \F_p$ be a prime field, and let $A \seq \F^n$ be a set of size $\abs{A} = \alpha \cdot \abs{\F}^n$, for some $\alpha \in (0,1]$. Then, there exists a subspace $V \seq \F^n$ of dimension $\dim(V) \geq n - O(\log^4(1/\alpha))$ such that for all $v \in V$ it holds that
    \begin{equation*}
        \Pr_{a_1,a_2,a_3 \in \F^n}[a_1, a_2 \in A, a_3, a_4 \in -A] \geq \Omega(\alpha^5)
        \enspace,
    \end{equation*}
    where $a_4 = v-a_1-a_2-a_3$.
    Furthermore, given a query access to the set $A$, there is an algorithm
    that runs in time $\exp(\log^4(1/\alpha)) \cdot \poly\log(1/\delta) \cdot \poly(n)$ and with probability $1-\delta$ computes a set of vectors $R \seq \F^n$ such that $V = \{v \in \F^n : \ip{v,r} = 0 \; \forall r \in R\}$.
\end{restatable}

We defer the proof of \cref{thm:bogolyubov-quasipoly} to \cref{sec:appendix-sanders}.

\subsection{Local correction lemma}
Using the probabilistic quasi-polynomial Bogolyubov-Ruzsa lemma (i.e., \cref{thm:bogolyubov-quasipoly}), for any vector $v \in V$ we can efficiently sample $x_1, x_2, x_3, x_4 \in X$ such that we can write
\begin{equation*}
    v = x_1 + x_2 - (x_3 + x_4) \;.
\end{equation*}

However, we need to be able to handle any vector $y \in \F^n$, and not just vectors in the subspace $V$. Towards that end, we show that since the subspace implied by the probabilistic quasi-polynomial Bogolyubov-Ruzsa lemma is of large dimension (i.e., of dimension $\dim(V) \geq n - O(\log^4(1/\alpha))$), we can decompose any vector $y\in\F^n$ as a linear combination of the form
\begin{equation*}
    y = x_1 + x_2 - (x_3 + x_4) + s
    \;,
\end{equation*}
where $x_1,x_2,x_3,x_4 \in X$ and $s \in \F^n$ is a \emph{sparse} vector. We stress that the sparsity of the decomposition is essential to our applications, as we cannot locally correct the shift part of the decomposition, and instead we need to compute it explicitly.

The above captures our local correction lemma which will be used throughout this paper.

\begin{lemma}[Efficient local correction]\label{cor:y=4x+s}
    Let $\F = \F_p$ be a prime field, and let $X \seq \F^n$ be a set of size $\abs{X} = \alpha \cdot \abs{\F}^n$, for some $\alpha \in (0,1]$. Then, there exists a non-negative integer $t \leq O(\log^4(1/\alpha))$, a~collection of $t$ vectors $B = \{b_1,\dots,b_t \in \F^n\}$, and $t$ indices $k_1,\dots, k_t \in [n]$ satisfying the following:
    \begin{description}
      \item Given a vector $y \in \F^n$, define $s = \sum_{j=1}^t \ip{y, b_j} \cdot \vec{e}_{k_j}$
            where $(\vec{e}_i)_{i \in [n]}$ is the standard basis.
            Then
            \begin{equation*}
                \Pr_{x_1,x_2,x_3 \in \F^n}[ x_1, x_2 \in X, x_3, x_4  \in -X] \geq \Omega(\alpha^5)
                \enspace,
            \end{equation*}
    \end{description}
    where $x_4=y - s - x_1 - x_2 - x_3$.
    
    Furthermore, suppose we have a randomized membership oracle $O_X$ that for every input $x \in \F^n$, computes the indicator $1_X(x)$ correctly with probability at least $2/3$. Then, there exists an algorithm that makes $\exp(\log^4(1/\alpha)) \cdot \poly\log(1/\delta) \cdot \poly(n)$ oracle calls to $O_X$, performs $\exp(\log^4(1/\alpha)) \cdot \poly\log(1/\delta) \cdot \poly(n)$ field operations, and with probability at least $1-\delta$ returns vectors $b_1,\dots,b_t \in \F^n$ and indices $k_1,\dots, k_t \in [n]$ as described above.
\end{lemma}

\begin{proof}
    Fix a set $X \seq \F^n$ of size $\abs{X} = \alpha \cdot \abs{\F}^n$ for some $\alpha \in (0,1]$.
    By applying \cref{thm:bogolyubov-quasipoly}, we obtain a subspace $V \seq \F^n$ of dimension $\dim(V) = n - t$ for $t = O(\log^4(1/\alpha))$.
    Let $R \subseteq \F_2^n \setminus \{\vec{0}\}$ be a set of vectors in $\F^n$ of size $t$ such that $V = \{v \in \F_2^n : \ip{v,r} = 0 \; \forall r \in R\}$. Indeed, we can let $R$ be a set of $t$ linearly independent vectors in $V^\perp$.
    
    By writing the vectors of $R$ in a matrix and diagonalizing the matrix,
    we obtain: (1) a set of vectors $B = \{b_1,\dots, b_t \in \F_2^n\}$ such that $\sp(B) = \sp(R)$,
    and (2) the corresponding pivot indices $k_1,\dots,k_t \in [n]$
    such that $b_j[k_j] = 1$ and $b_j[k_{j'}] = 0$ for all $j \neq j'$.

    Given a vector $y \in \F^n$, define $s = \sum_{j=1}^t \ip{y, b_j} \cdot \vec{e}_{k_j}$, where $(\vec{e}_i)_{i \in [n]}$ is the standard basis,
    and let $v = y-s$.
    It is straightforward to verify that $v \in V$.
    Then for any $j \in [t]$ we have
    \begin{equation*}
        \ip{v, b_j}
        = \ip{y, b_j}  - \sum_{j=1}^t c_j \cdot \ip{\vec{e}_{k_j},b_j}
        \overset{\text{(*)}}{=} \ip{y, b_j}  - c_j \cdot \ip{\vec{e}_{k_j},b_j}
        \overset{\text{(**)}}{=} \ip{y, b_j}  - \ip{y, b_j} = 0
        \enspace,
    \end{equation*}
    where (*) is because $\ip{\vec{e}_{k_{j'}},b_j} = b_j[i_{j'}] = 0$ for $j \neq j'$,
    and (**) is because $\ip{\vec{e}_{k_j},b_j} = b_j[i_{j}] = 1$.
    
    Now, since $v \in V$, by the guarantees of \cref{thm:bogolyubov-quasipoly} it follows that 
    \begin{equation*}
        \Pr_{x_1,x_2,x_3 \in \F^n}[ x_1 \in X, x_2 \in X, x_3 \in -X,v - x_1 - x_2 - x_3 \in -X] \geq \Omega(\alpha^5)
        \enspace.
    \end{equation*}
    
    For the furthermore part, note that we can boost the success probability of the membership oracle $O_X$. That is, given a query $x$ we can decide if $x \in X$ with confidence $1 - \frac{1}{t}$ be repeatedly calling it $O(\log(t))$ times and taking the majority vote. In particular, for $t = \exp(\log^4(1/\alpha)) \cdot \poly\log(1/\delta) \cdot \poly(n)$, by making $O(\poly\log(1/\alpha) + \log(1/\delta) + \log(n))$ calls to $O_X(x)$ for each element $x \in \F^n$ we need to query, we may assume that all queries output whether $x \in X$ or not correctly.

    The furthermore part of the lemma follows immediately from the computational guarantees of \cref{thm:bogolyubov-quasipoly} together with the diagonalization procedure described above.
\end{proof}

\section{Worst-case to average-case reductions for matrix multiplication}\label{sec:MM-reduction}
We prove the worst-case to average-case reduction for matrix multiplication problem in this section. Let's restate
\cref{thm:mm-alg-reduction-intro} below.

\mmreduction*

We divide the proof into two parts, namely for when $\abs{\F}\geq \alpha/2$ and when $\abs{\F}\leq \alpha/2$.
The proof for the former case is given in \cref{sec:mm-large-fields} and the proof of the latter is given in \cref{sec:mm-small-fields}.

We would like to point out that when the field size is large enough, we can use the standard interpolation
techniques for low-degree polynomials to prove the reduction. However, the problem becomes more challenging
when the field size is small (say, $\F = \F_2$), and showing the reduction in this case requires novel ideas. We will discuss 
both cases in detail in the following sections.

\subsection{Reduction for matrices over small fields}\label{sec:mm-small-fields}
In this section, we show a worst-case to average-case reduction for matrix multiplication problem over small fields, namely, where $\abs{\F} \leq 2/\alpha$, where $\alpha$ is the success rate of the average-case algorithm.
Informally, we will show that if there exists an algorithm that is able to compute the multiplication for a small
percentage of matrices, then it is possible to boost this algorithm such that it works for all matrices,
without sacrificing the running time too much.
Before we proceed with the formal result, we need the following lemma known as Freivalds' algorithm.

\begin{lemma}[Freivalds' Algorithm~\cite{freivalds1977probabilistic}]\label{lemma:freivalds}
    Given matrices $A, B, C \in \F^{n \times n}$ there exist a probabilistic algorithm that verifies whether $A\cdot B = C$ with
    failure probability $2^{-k}$ where the algorithm runs in $O(kn^2)$.
\end{lemma}

Throughout the proof, we will use Freivalds' algorithm to verify the result of matrix multiplication instances, 

In particular, it suffices to design an algorithm that given two matrices $A,B$ outputs their product with some non-negligible probability $\eps>0$. By repeating the algorithm $O(1/\eps)$ times, we can boost the probability of outputting the correct answer to a constant arbitrarily close to 1. This is done by applying Freivalds' algorithm on each of the outputs of the algorithm, rejecting incorrect outputs with high probability, and accepting when the correct answer is found.

We now demonstrate the main result of this section, which corresponds to the first case in \cref{thm:mm-alg-reduction-intro}.
\begin{theorem}\label{thm:mm-alg-reduction}
    Let $\F = \F_p$ be a prime field, $n \in \N$, and $\alpha\coloneqq\alpha(n) \in (0,1]$.
    Suppose that there exists an algorithm $\ALG$ that, on input two matrices $A, B \in \F^{n \times n}$ runs in time $T(n)$
    and satisfies
    \begin{equation*}
        \Pr[\ALG(A,B)=A\cdot B] \geq \alpha \,,
    \end{equation*}
    where the probability is taken over the random inputs $A,B\in \F^{n \times n}$ and the randomness of~$\ALG$.
 If $\abs{\F} \leq 2/\alpha$, then there exists a randomized algorithm $\ALG'$ that for \emph{every} input $A, B \in \F^{n \times n}$ and $\delta>0$, runs in time $O(\frac{\exp(O(\log^5(1/\alpha)))}{\delta} \cdot T(n))$ and outputs $AB$ with probability at least $1-\delta$.

\end{theorem}
    Below we will prove the theorem assuming the algorithm~$\ALG$ is deterministic, but a straightforward generalization of the proof works for randomized algorithms as well.
    To proceed with the proof, we first need the following definitions.
    \begin{definition}\label{def:good-matrices}
    Let $X$ be the set of matrices $A$ such that $\ALG$ computes their product with matrices $B$ with probability at least $\alpha/ 2$. More formally
    \begin{equation*}
        X = \{A: \Pr_{B}[\ALG(A,B)=A\cdot B]\geq \alpha / 2\} \enspace.
    \end{equation*}
    Similarly, for each $A \in \F^{n \times n}$, we define $Y_A$ to be the set of matrices $B$ such that given $A$ and $B$, $\ALG$ correctly computes $A \cdot B$. In other words
    \begin{equation*}
        Y_A = \{B: \ALG(A,B)=A\cdot B\} \enspace.
    \end{equation*}
    \end{definition}
    \begin{claim}
    $X$ and $Y_A$, where $A \in X$, have density at least $\alpha/2$.
    \end{claim}
    \begin{proof}
    Let $P_A$ be the random variable $P_A := \Pr_{B}[\ALG(A,B)=A\cdot B]$. From the definition, it is clear that $\E_A[P_A] \geq \alpha$.
    Now, by contradiction, if $\Pr_A[P_A \geq \alpha/2] < \alpha/2$ then we have
    \begin{equation*}
        \E_A[P_A] = \E_A{\Pr_{B}[\ALG(A,B)=A\cdot B]} < \alpha/2 \cdot 1 + (1-\alpha/2) \cdot \alpha/2 < \alpha \,.
    \end{equation*}
    Hence, $\Pr[P_A \geq \alpha/2] \geq \alpha/2$ and $X$ has density greater than or equal to $\alpha/2$.
    It follows from the definitions of $X$ and $Y_A$ that for all $A \in X$, $Y_A$ has density at least $\alpha / 2$.
    \end{proof}

    Next, we need the following definition.

    \begin{definition}\label{def:matrices-M-and-N-2}
       For a matrix $A \in \F^{n \times n}$ and $k \in [n]$, let $M^{k}_A = A +L^{k}_A$ where we define the matrix $L^{k}_A$ as follows.

      \begin{enumerate}
          \item First, choose a random subset $S$ of size $k$ from $[n]$.
          \item For each $i \in S$ let the $i$'th row of $L^{k}_A$ be uniformly random in $\F^n$.
          \item For all $j \in [n] \setminus S$, let the $j$'th row of $L^{k}_A$ be a random linear combination of the rows in $S$.
      \end{enumerate}

    \end{definition}
    \begin{remark}
    For matrix $L^{k}_A$ we have that $\rk(L^{k}_A)\leq k$, because every row indexed by $j \in [n] \setminus S$ is a linear combination of the rows in $S$.
    \end{remark}
    \begin{remark}
        If the random rows in $S$ are not linearly independent, we can throw them away and repeat Step 2.
        It is not hard to see that this event happens only with constant probability, and we can
        check this in $O(nk^2)$ time.
    \end{remark}
    
    The following lemma shows that matrix $M^{2k}_A = A + (L^{2k}_A)$ belongs to any subspace of matrices of constant co-dimension $k$ with constant probability.

    \begin{lemma}\label{lemma:extended-matrix-2}
        Given a matrix $A \in \F^{n \times n}$ and $k \in [n]$, for any subspace $V \seq \F^{n \times n}$ of $\dim(V) \geq n-k$ we have
        \begin{equation*}
            \Pr[M^{2k}_A \in V] \geq \frac{1}{2\abs{\F}^k}\enspace.
        \end{equation*}
    \end{lemma}

    \begin{proof}
        Since $V$ has co-dimension $k$, it can be defined by $k$ linear constraints on the elements of the matrix as follows.
        \begin{equation*}
            M^{2k}_{A}(i_0,j_0) + M^{2k}_{A}(i_1,j_1) + ... + M^{2k}_{A}(i_r,j_r) = 0 \enspace,
        \end{equation*} 
        where $r \in [1, n^2]$ denotes the number of coordinates that this constraint depends on.
        By re-writing $M^{2k}_A$ as a vector $\mathbf{m} \in \F^{n^2}$,
        we can construct the system of equations $G\cdot \mathbf{m} = \mathbf{0}$ for membership in $V$. Here, $G$ denotes the matrix of size $k \times n^2$, where each row specifies one single constraint of the aforementioned form. Now, if we diagonalize $G$ using Gaussian elimination, we can re-write the system of equations in the form $G'\cdot \mathbf{m} = \mathbf{0}$ for a matrix $G'$, where for each row $a$ in $G'$, there exists a column $b_a$ which has value 1 in this row and 0 in the other rows.
        
        For all $b_a$ where $a \in [k]$, we consider the coordinate $\mathbf{m}_{b_a}$. The set of these $k$ coordinates $\{\mathbf{m}_{b_1},\mathbf{m}_{b_2},...,\mathbf{m}_{b_k}\}$ corresponds to $k$ pairs of coordinates $\{(c_1, c'_1),(c_2, c'_2),...,(c_k, c'_k)\}$ in the original matrix. Note that these $k$ coordinates belong to at most $k$ rows in $M^{2k}_A$, and we want to bound the probability that none of these rows in $L^{2k}_A$ is a linear combination of the other rows. Let $Z$ be the event that all the $2k$ rows are pairwise linearly independent in $L^{2k}_A$. We have
        \begin{align*}
            \Pr[Z] &= \left(1 - \frac{1}{2^{2k}}\right)\left(1 - \frac{2}{2^{2k}}\right)\left(1 - \frac{4}{2^{2k}}\right) ... \left(1 - \frac{2^{k-1}}{2^{2k}}\right) \\
            &\geq \left(1 - \frac{2^{k-1}}{2^{2k}}\right)^{k} \geq \left(1 - \frac{1}{2^{k+1}}\right)^{k} \\
            &\geq 1 - \frac{k}{2^{k+1}} \geq \frac{1}{2} \enspace.
        \end{align*}

        Now, we observe that if $Z$ happens, it means that for all $k$ constraints, there exist a coordinate which we denote by $(c_i, c'_i)$ in $M^{2k}_A$ such that none of the other constraints depend on the value of $M^{2k}_A(c_i, c'_i)$, and the value of $M^{2k}_A(c_i, c'_i)$ is chosen uniformly at random. Hence, this random value is equal to the unique solution which satisfies $i$th constraint (assuming values of all other coordinates involved in this constraint are determined beforehand) with probability $1/\abs{\F}$. Therefore, the probability that $M^{2k}_A \in V$ is bounded by
        \begin{align*}
            \Pr[M^{2k}_A \in V] &= \Pr[\text{All $k$ constraints are satisfied}] \\ 
            &\geq \Pr[Z] \cdot \frac{1}{\abs{\F}^k} \\
            &= \frac{1}{2\abs{\F}^k}
            \enspace. \qedhere
        \end{align*}
    \end{proof}

\begin{proof}[Proof of \cref{thm:mm-alg-reduction}]
    To prove this theorem, we design the following algorithm and prove that this algorithm outputs the correct answer for the matrix multiplication problem with high probability.
\\
\begin{tcolorbox}[title=Algorithm~\customlabel{algone}{1}: Matrix multiplication reduction over small fields]   
\paragraph{Input: $\ALG$, $A, B \in \F^{n \times n}$}
\paragraph{Output: $A \cdot B$}
Set $k$ to be $O(\log^4(1/\alpha))$.
\begin{enumerate}
    \item Set $k$ to be $O(\log^4(1/\alpha))$.
    \item For matrices $A$ and $B$, construct the matrices $M^{2k}_A$ and $M^{2k}_B$.
    \item Sample 3 random matrices $R_1, R_2, R_3 \in \F^{n \times n}$ and set $R_4 = R_1 + R_2 - R_3 - A - M^{2k}_A$ so that $A + M^{2k}_A = R_1+R_2-R_3-R_4$.
    \item Sample 12 random matrices $S^{(t)}_1, S^{(t)}_2, S^{(t)}_3 \in \F^{n \times n}$ and set $S^{(t)}_4 = S^{(t)}_1 + S^{(t)}_2 - S^{(t)}_3 - B - M^{2k}_B $ for $t \in \{1,2,3,4\}$, so that $B + M^{2k}_B = S^{(t)}_1 + S^{(t)}_2 - S^{(t)}_3  - S^{(t)}_4$.
    \item Compute $O_L = \sum_{t=1}^4{\sum_{s=1}^4 {\mathsf{sign}}_{t,s}\ALG(R_t, S^{(t)}_s)}$, where $\mathsf{sign}_{t,s} = -1$ if $\{t,s\} \cap \{1,2\} = 1$, and $\mathsf{sign}_{t,s} = 1$ otherwise.
    \item Compute $O = O_L - A \cdot L^{2k}_B - L^{2k}_A \cdot B - L^{2k}_A \cdot R^{2k}_B$.\label{line:get-rid-of-noise}
    \item {If} $O = A \cdot B$ (check using \cref{lemma:freivalds}), {then} {return} $O$.
\end{enumerate}
\end{tcolorbox}

    \paragraph{Correctness:}
    Let's consider $X$ as it is defined in \cref{def:good-matrices}. By applying \cref{thm:bogolyubov-quasipoly} on $X$, we can conclude that there exists subspace $V_X$ with co-dimension at most $O(\log^4(1/\alpha))$ and the guaranteed properties.
    On the other hand, by \cref{lemma:extended-matrix-2} we have that
    \begin{equation*}
        \Pr[M^{2k}_A \in V_X] \geq \frac{1}{2\abs{\F}^{O(\log^4(1/\alpha))}} \enspace.
    \end{equation*}
    Assuming $M^{2k}_A \in V_X$, by \cref{thm:bogolyubov-quasipoly}, we have that
    \begin{equation*}
        \Pr_{R_1,R_2,R_3}[R_1,R_2,-R_3,-R_4 \in X] \geq \Omega(\alpha^5) \enspace.
    \end{equation*}
    Now, for each $R_t$ with $t \in \{1,2,3,4\}$, we can consider $Y_{R_t}$ using \cref{def:good-matrices}. Similarly, since each $Y_{R_t}$ has density at least $\alpha/2$, we can apply \cref{thm:bogolyubov-quasipoly} on $Y_{R_t}$ to define the corresponding subspaces $V_{Y_{R_t}}$. Having these four subspaces, we define $V_Y = V_{Y_{R_1}} \cap V_{Y_{R_2}} \cap V_{Y_{-R_3}} \cap V_{Y_{-R_4}}$. It is not hard to see that since each of the four subspaces has co-dimension $O(\log^4(1/\alpha))$, the co-dimension of $V_Y$ is at most $4\cdot O(\log^4(1/\alpha))$. Thus, by \cref{lemma:extended-matrix-2}
    \begin{equation*}
        \Pr[M^{2k}_B \in V_Y] \geq \frac{1}{2\abs{\F}^{O(\log^4(1/\alpha))}}\enspace.
    \end{equation*}
    Given $M^{2k}_B \in V_Y$, for each $t \in \{1,2,3,4\}$ by \cref{thm:bogolyubov-quasipoly}
    \begin{equation*}
        \Pr_{S^{(t)}_1, S^{(t)}_2, S^{(t)}_3 }[S^{(t)}_1, S^{(t)}_2, -S^{(t)}_3 ,-S^{(t)}_4 \in Y_{R_t}] \geq \Omega(\alpha^5) \enspace.
    \end{equation*}
    It is important to note that since $L^{2k}_A$ and $L^{2k}_B$ have rank less than or equal to $2k$, and all linear combinations of their rows are known previously, we can compute the multiplications in Step \ref{line:get-rid-of-noise} in time $O(n^2 \cdot \log^4(1/\alpha))$.

    Since all the events defined above are independent, we conclude that the algorithm succeeds with the following probability
    \begin{align*}
        \Pr[\mathrm{Algorithm~\ref{algone} \ succeeds}] &\geq \frac{\Omega(\alpha^{25})}{O(\abs{\F}^{O(\log^4(1/\alpha)))}}
        \geq \frac{\Omega(\alpha^{25})}{O(\frac{10}{\alpha})^{O(\log^4(1/\alpha))}}
        \geq \exp(-\log^5(1/\alpha))
        \enspace.
    \end{align*}
    
    Therefore, by repeating the algorithm $\frac{\exp(\log^5(1/\alpha))}{\delta}$ times, and using Freivalds' algorithm for verification, we obtain an algorithm that solves the matrix multiplication on all instances with probability at least $1-\delta$. 
    
\paragraph{Running time:}
Th running time of the procedure described above is essentially dominated by $\frac{\exp(\log^5(1/\alpha))}{\delta}$ calls to the weak average case algorithm, and hence the total running time is  $\frac{\exp(\log^5(1/\alpha))}{\delta} \cdot T(n)$.
In particular, even if the algorithm succeeds on a sub-constant fraction of inputs $\alpha = \exp(\log^{0.199}(n))$, the reduction turns it into an algorithm that works for  worst case instances in time $T(n) \cdot n^{o(1)}$.
\end{proof}

\subsection{Reduction for matrices over large fields}\label{sec:mm-large-fields}
We now prove case 2 of \cref{thm:mm-alg-reduction-intro} in this section. For concreteness, 
we restate this result.

\begin{theorem}\label{thm:wc-to-ac-mm-large}
    Let $\F = \F_p$ be a prime field, $n \in \N$, and $\alpha\coloneqq\alpha(n) \in (0,1]$.
    Suppose that there exists an algorithm $\ALG$ that, on input two matrices $A, B \in \F^{n \times n}$ runs in time $T(n)$
    and satisfies
    \begin{equation*}
        \Pr[\ALG(A,B)=A\cdot B] \geq \alpha \,,
    \end{equation*}
    where the probability is taken over the random inputs $A,B\in \F^{n \times n}$ and the randomness of~$\ALG$.
 If $\abs{\F} \geq 2/\alpha$, then there exists a randomized algorithm $\ALG'$ that for \emph{every} input $A, B \in \F^{n \times n}$ and $\delta>0$,
        runs in time $O(\frac{1}{\delta \cdot \alpha^4} \cdot T(n))$ and outputs $AB$ with probability at least $1-\delta$.
\end{theorem}
\begin{proof}
    To prove \cref{thm:wc-to-ac-mm-large}, we use the following algorithm and we prove that this algorithm outputs the correct answer for the matrix multiplication problem with high probability.
\\
\begin{tcolorbox}[title=Algorithm~\customlabel{algtwo}{2}: Matrix multiplication reduction over large fields]   
\paragraph{Input: $\ALG$, $A, B \in \F^{n \times n}$}
\paragraph{Output: $A \cdot B$}
Set $k$ to be $O(\log^4(1/\alpha))$.
\begin{enumerate}
    \item Let $X$ and $Y$ be matrices chosen uniformly at random from $\F^{n \times n}$, and let $i, j,$ and $k$ be chosen uniformly
        at random from $\F$.\label{line:line-sampling}
    \item Compute $\ALG(A+iX, B+iY), \ALG(A+jX, B+jY),$ and $\ALG(A+kX, B+kY)$.\label{line:matrix-sampling} 
    \item {If} the computations are correct (check using \cref{lemma:freivalds}), {then} compute $A \cdot B$ by interpolating $(A+iX, B+iY), (A+jX, B+jY)$, and $(A+kX, B+kY)$.
\end{enumerate}
\end{tcolorbox}

\paragraph{Correctness:}

Again, we prove the result assuming the algorithm~$\ALG$ is deterministic, but a straightforward generalization of the proof works for randomized algorithms as well. We define the set of good pairs of matrices for $\ALG$ as follows.
\begin{definition}\label{def:good-pair-matrices}
Let $S \subseteq (\F^{n \times n} \times \F^{n \times n})$ be the set of pairs of matrices such that for all $(M, N) \in S$ we have $\ALG(M, N) = M \cdot N$. More formally
\begin{equation*}
    S = \{(M, N): \ALG(M,N) = M\cdot N, M \in \F^{n \times n}, N \in \F^{n \times n}\} \enspace.
\end{equation*}
Note that by definition, $S$ has density at least $\alpha$.
\end{definition}

\begin{claim}\label{claim:good-pair-density}
Let $\ell_{X,Y} = \{(A+iX, B+iY): i \in \F\}$ be the line that passes through $(A,B)$ and is defined by matrices $X$ and $Y$. Then, with probability $\alpha/2$, at least $\alpha/2$ fraction of pairs of matrices on this line belong to $S$.
\end{claim}
\begin{proof}
    Let $P_{(X,Y)}$ be the random variable $P_{(X,Y)} = \Pr_{i}[\ALG(A+iX,B+iY)=(A+iX)\cdot (B+iY)]$. From the definition,  $\E[P_{(X,Y)}] \geq \alpha$.
    Now, by contradiction, if $\Pr[P_{(X,Y)} \geq \alpha/2] < \alpha/2$ then we have
    \begin{equation*}
        \E[P_{(X,Y)}] = \E_{X,Y}{\Pr_{i}[\ALG(A+iX,B+iY)=(A+iX)\cdot (B+iY)]} < \alpha/2 \cdot 1 + (1-\alpha/2) \cdot \alpha/2 < \alpha \enspace.
    \end{equation*}
    Hence, $\Pr[P_{(X,Y)} \geq \alpha/2] \geq \alpha/2$.
\end{proof}

Having \cref{def:good-pair-matrices}, we can make the following claim.
\begin{claim}
The three pairs of matrices defined in Step~\ref{line:matrix-sampling} of Algorithm~\ref{algtwo} belong to $S$ with probability at least $\alpha^4 /16$.
\end{claim}

\begin{proof}
Note that in Step \ref{line:line-sampling} of Algorithm \ref{algtwo}, we are sampling a line $\ell_{X,Y}$, which by \cref{claim:good-pair-density} has density $\alpha/2$ of good pair of matrices with probability $\alpha/2$.
Also, in Step \ref{line:matrix-sampling}, we are sampling three uniformly random pairs on this line. Assuming $\ell_{X,Y}$ is a line with density $\alpha/2$ of good pairs of matrices, with probability at least $(\alpha/2)^3$ these three pairs belong to $S$. Hence, total probability that we sample three pairs such that they all belong to $S$ is at least $(\alpha/2)^4 = \alpha^4 /16$.
\end{proof}
Since matrix multiplication is a polynomial of degree 2 in the entries of the matrices, having 3 pairs where $\ALG$ outputs correct answers enables us to interpolate the value of $A \cdot B$. Thus, the algorithm succeeds with probability $O(\alpha^4)$.
By repeating the algorithm and verifying the answer using Freivalds' algorithm, one can amplify the success probability to any arbitrary constant.

\end{proof}

\section{Worst-case to average-case reductions for online matrix-vector multiplication}
\label{sec:OMV}
In this section we prove \cref{thm:MV-weak-avg-case}, which we restate below.

\MV*

\begin{remark}[Large fields]
    We would like to point out that this problem is more interesting when the field $\F$ is small. Indeed, if the size of $\F$ is relatively large
    (say, $\abs{\F} >2/\alpha$), then we can think of the matrix-vector multiplication as a polynomial of degree at most one in the elements of the vector. Therefore, we can use the standard self-correction techniques
    for evaluating low-degree polynomials to solve this problem. In particular, given a query $v \in \F^n$ we can sample a line $\ell \in \F^n$ that passes through $v$, and query two random vectors that belong to $\ell$, and compute $Mv$ by interpolating the two queried points.

\end{remark}

Before proceeding with the formal proof of \cref{thm:MV-weak-avg-case}, we informally outline the argument.
\paragraph{Proof sketch:}
The proof of \cref{thm:MV-weak-avg-case} relies on \cref{cor:y=4x+s}, which shows that any vector in $\F^n$ can be self-corrected via a linear constraint involving four vectors $x_1, x_2, x_3, x_4 \in X$ and a shift by sparse vector $u$.  This result will be used several times in the proof of \cref{thm:MV-weak-avg-case}, which we explain below.
\begin{itemize}
\item
First, note that if $\Pr_{M,v}[\DS_M(v) = Mv] \geq \alpha$, then there is a collection $Z \seq \F^n$ of size $\abs{Z} \geq \alpha/2 \cdot \abs{\F}^n$ of \emph{good} matrices, i.e., those matrices $M$ 
on which $\DS$ succeeds to compute $Mv$ on some non-negligible fraction of vectors $v$. More formally, $Z = \{M \in \F^{n \times n} : \Pr_{v}[\DS_M(v) = Mv] \geq \alpha/2 \}$.

\item
First time we apply \cref{cor:y=4x+s} on the set $Z \seq \F^{n \times n}$.
(That is, we identify $n \times n$ matrices with vectors of length $N =n^2$.)
Roughly speaking, given an arbitrary matrix $M$ we will apply the lemma so that we can write
$M = M_1 + M_2 - M_3 - M_4 + U$, where $M_1,M_2,M_3, M_4 \in Z$ and $U$ is a sparse matrix.
By the assumption in \cref{thm:MV-weak-avg-case} each $M_i$ succeeds on a non-negligible fraction of vectors $v$.

\item
Second time \cref{cor:y=4x+s} will be used with the set $X = X_{M_i}$ of vectors $v \in \F^n$ on which $\DS$ outputs $M_i v$ correctly,
where $M_i$ is each of the matrices above.
Using the lemma we will be able to represent every vector $v \in \F^n$
as $v = x_1 + x_2 - x_3 - x_4 + u$, where the $x_j$'s belong to $X_{M_i}$, i.e., the data structure outputs $M_i x_j$ correctly for all $j=1,2,3,4$, and $u \in \F^n$ is a sparse vector.

\item
In particular, for each of the matrices $M_i$ the data structure computes correctly $M_i x_j$ for all $j=1,2,3,4$,
and $M_i u$ can be computed in the query phase by reading only $O(1)$ columns of $M_i$, as $u$ is a sparse vector.
\end{itemize}

Before proceeding with the formal proof of \cref{thm:MV-weak-avg-case}, we need the following definition of small-bias sample spaces, and the theorem regarding their existence.

\begin{definition}[Small-bias sample spaces]\label{def:small-bias-space}
A sample space $S$ over $\F^n$ is called \emph{$\eps$-biased}
if for every $r \in \F^n \setminus \{0\}$ and every $b \in \F$ it holds that
\begin{equation*}
  \abs{\Pr_{s \in S}[\ip{s,r} = b] - \frac{1}{\abs{\F}}} \leq \eps
  \enspace.
\end{equation*}  
\end{definition}

In other words, a sample space is $\eps$-biased if it $\eps$-fools every nontrivial linear test,
i.e., for any $r \in \F^n \setminus \{0\}$
the distribution of $\ip{s,r}$ with $s$ sampled from $S$ is close in distribution to $\ip{s,r}$ for a \emph{uniformly random $s \in \F^n$}.
When the exact value of $\eps$ is not important, e.g., by setting $\eps = 0.1$,
we usually call such sample spaces \emph{small-bias sets}.
These objects have been introduced in the work of Naor and Naor \cite{NN93},
followed by a long line of work culminating in the recent
almost optimal construction of Ta-Shma~\cite{TaShma17}, who showed an efficient construction
of such sets of size $O(n/\eps^{2+o(1)})$ over $\F_2$. For our purposes, even a randomized construction will be sufficient (see, e.g., Corollary~3.3 in~\cite{azar1998approximating}).
\begin{theorem}\label{thm:small-biased-sets}
For every finite field $\F$, constant~$\eps\in[0,1]$ and $n \in \N$, a random set $S \seq \F^n$ of size $O(n \log{|\F|})$  is an $\eps$-biased with high probability. For the field $\F=\F_2$, there exists an explicit construction of size $\abs{S} = O(n)$.
\end{theorem}

\noindent For concreteness, we will take $\eps = 0.1$, which suffices for our application.

We are now ready to prove \cref{thm:MV-weak-avg-case}.

\begin{proof}[Proof of \cref{thm:MV-weak-avg-case}]
    For each matrix $M \in \F^{n \times n}$, let $\DS_M$ be the weak average-case data structure implied by the hypothesis of the theorem, and denote by $X_M = \{ x \in \F^n \colon \DS_M(x)=Mx\}$ the set of vectors 
    on which the data structure outputs the correct answer.
    Let $Z \seq \F^{n \times n}$ be the set of matrices on which the data structure outputs the correct answer on at least $\frac{\alpha}{2}$-fraction of the inputs; that is, $Z = \{ M \in \F^{n \times n} \colon \abs{X_M} \geq \frac{\alpha}{2} \abs{\F}^n \}$. By Markov's inequality, we have $|Z| \geq \frac{\alpha}{2} \cdot \abs{\F}^{n^2}$. Observe that given access to $\DS_M$, we can easily construct a probabilistic oracle for approximate membership in $Z$.

    \begin{claim}
        \label{clm:membership_oracle}
        There exists a probabilistic membership oracle $O_Z$ that for any query $M \in \F^{n \times n}$ makes $t = O(1/\alpha)$ calls to $\DS_M(x)$ for uniformly random $x \in \F^n$, compares the result to $Mx$,
        and accepts if and only if at least $\alpha/3$ fraction of the calls to $\DS_M(x)$ output the correct answer. The oracle  has the following guarantees.
        \begin{itemize}
            \item If $M \in Z$, then $\Pr[O_Z(M) = ACCEPT]>2/3$.
            \item If $\abs{X_{M}} \leq \frac{\alpha}{4} \abs{\F}^n$ then $\Pr[O_Z(M) = REJECT]>2/3$.
        \end{itemize}
    \end{claim}

    Using the weak average-case data structure $\DS_M$, we construct a worst-case data structure $\DS'$ as follows. First, we describe the preprocessing stage of the data structure $\DS^{(M)}$.

\begin{tcolorbox}[title=Preprocessing:]
\paragraph{Input:} A matrix $M \in \F^{n \times n}$ 
    \begin{enumerate}
      \item \textsf{Self-correcting $M$:} Using \cref{cor:y=4x+s} with probability $1-\delta$ we represent the matrix $M$ as ${M = M_1 + M_2 - M_3 - M_4 + U}$, where each $M_i$ has a large $X_{M_i}$ and $U$ is a $t$-sparse matrix for $t=O(\log^4(1/\alpha))$. 
      The running time of this step is $\exp(\log^4(1/\alpha)) \cdot \poly\log(1/\delta) \cdot \poly(n)$

      \item \textsf{Self-correcting $x$:} Then, for each $M_i$, we apply \cref{cor:y=4x+s} on $X_{M_i} = \{x \in \F^n \colon \DS_{M_i}(x) = M_ix\}$, and compute a collection of $t \leq O(\log^4(1/\alpha))$ vectors $B_i = \{b^{(i)}_1,\dots,b^{(i)}_t \in \F^n\}$
      and $t$ indices $k^{(i)}_1,\dots, k^{(i)}_h \in [n]$ that allow us to represent each vector
      $x = x_1+x_2-x_3-x_4 + u_i$, where $u_i$ has at most $O(\log^4(1/\alpha))$ non-zero elements, and $x_j \in X_{M_i}$ for all $i=1,2,3,4$.

      \item Let $S \seq \F^n$ be a small-biased set obtained by taking $O(n)$ uniformly random vertices in $\F^n$. Note that for $\F = \F_2$ we can take the explicit set $S$ from \cref{thm:small-biased-sets}
      with $\eps = 0.1$.
      
      \item For each $e \in S$ compute the multiplication from the left $e M_i$ of $e$ with each of the $M_i$, and store the pairs $(e,e M_i)$ in the data structure.
    \end{enumerate}
\end{tcolorbox}

   Overall, the data structure stores the following information in the preprocessing step:
    \begin{itemize}
    
        \item The sparse matrix $U \in \F^{n \times n}$ with at most $O(\log^4(1/\alpha))$ non-zero entries, obtained in Step~1;
        
        \item For each $M_i$, the weak average-case data structure $\DS_{M_i}$ for $M_i$, which outputs $\DS_{M_i}(x) = M_i \cdot x$ correctly on at least $\alpha/4$ fraction of $x$'s.
        
        \item For each $M_i$, the corresponding $t = O(\log^4(1/\alpha))$ vectors $B_i = \{b^{(i)}_1,\dots,b^{(i)}_t \in \F^n\}$, and $t$ indices $k^{(i)}_1,\dots, k^{(i)}_t \in [n]$, obtained in Step~2;

        \item The pairs $(e,e M_i)$ for every vector $e$ in the small-biased set $S \seq \F^n$ and for every $M_i$, obtained in Step~3.
    \end{itemize}
    
    \paragraph{Preprocessing time:} The preprocessing time is determined by the time needed to find matrices $M_i$, and the preprocessing time of the weak-average-case data structures. It is easy to verify that the preprocessing time is bounded by
    $4p + \exp(\log^4(1/\alpha)) \cdot \poly\log(1/\delta)\cdot \poly(n)$, where $p$ is the preprocessing time of the weak-average-case data structure. 

    \paragraph{Memory used:} In the preprocessing step, we store
    (1) 4 weak-average-case data structures of size $s$ for each of the matrices $M_i$,
    (2) for each matrix $M_i$ we store the collection of $t = O(\log^4(1/\alpha))$ vectors $B_i$ and $t$ indices,
    (3) a representation of the sparse matrix $U$ using $O(\log^4(1/\alpha) \log(n))$ field elements.
    (4) The pairs $(e,eM_i)$ for every vector $e$ in the small-biased set $S$, which is of size $O(n)$. Hence, the total space used is $4s + O(\log^4(1/\alpha) n) + O(n^2)$.

\medskip

    Next we describe the query phase of the worst-case data structure. Recall, for each matrix $M_i$ where $i=1,2,3,4$, we store the vectors $B_i = \{b^{(i)}_1,\dots,b^{(i)}_h \in \F^n\}$
        and indices $k^{(i)}_1,\dots, k^{(i)}_h \in [n]$ in out data structure,
        to compute $u_i$ every vector in $\F^n$ can be written as a linear combination of four vectors in $X_{M_i}$.
        The query phase works as follows.

\begin{tcolorbox}[title=Query phase:]
\paragraph{Input:} A query $x \in \F^n$
    \begin{enumerate}
      \item
        For $i \in \{1,2,3,4\}$,
        sample random $x^{(i)}_1,x^{(i)}_2,x^{(i)}_3 \in \F^n$
        and let $x^{(i)}_4$ be such that $x=u_i+x^{(i)}_1+x^{(i)}_2-x^{(i)}_3-x^{(i)}_4$.
      \item
         For each matrix $M_i$ and for $j\in \{1,2,3,4\}$, apply $\DS^{(M_i)}_\alpha$ to $x^{(i)}_j$.
      \item
        Verify that $\DS^{M_i}(x^{(i)}_j) = M_i x^{(i)}_j$ using the small biased set $S$.
        Specifically, sample $O(\log(1/\delta))$ vectors $e \in S$,
        and check that 
        \begin{equation*}
            \ip{e, \DS_{M_i}(x^{(i)}_j)} = \ip{eM_i, x^{(i)}_j}\enspace.
        \end{equation*}
        If the answer of $\DS_{M_i}(x^{(i)}_j)$ outputs the correct answer, then the inner products will all be equal.
        If $\DS_{M_i}(x_j) \neq M_i x^{(i)}_j$, then a random $e \in S$ will catch an inequality with probability at least $0.4$. 
      \item
        By repeating the sampling above for $O(\log(1/\delta) \cdot 1/\alpha^5)$ times, for each
        $i \in \{1,2,3,4\}$ we will find such
        $x^{(i)}_1,x^{(i)}_2,x^{(i)}_3,x^{(i)}_4$ on which $\DS_{M_i}(x^{(i)}_j)$ outputs the correct answer with high probability.
      \item
        Compute $M_i u_i$ directly. Since $u_i$ has at most $O(\log^4(1/\alpha))$ non-zero coordinates, it follows that $M_i u_i$ can be computed in time $O(\log^4(1/\alpha) n)$.
      \item
        For $i \in \{1,2,3,4\}$ compute $M_i y$ by taking $M_i x^{(i)}_1 + M_i x^{(i)}_2 - M_i x^{(i)}_3 - M_i x^{(i)}_4 + M_i u_i$.
      \item
        Compute $U y$ directly. Since $U$ has at most $O(\log^4(1/\alpha))$ non-zero elements, this can be done in time $O(\log^4(1/\alpha) \cdot \log(n))$.
      \item
        Return $My = M_1 y + M_2 y - M_3 y - M_4 y + U y$.
    \end{enumerate}
\end{tcolorbox}

    \paragraph{Correctness:} To prove the correctness, we bound the failure probability of the algorithm.
    Note that the failure of the algorithm only depends on the verification procedure in Step~3. In other words,
    if $\DS'_{M}(x) \neq Mx$, then at least for one pair of $(i, j)$ we have that
    $\DS_{M_i}a(x^{(i)}_j) \neq M_i x^{(i)}_j$, and none of the sampled vectors $e$ has caught this inequality.
    On the other hand, this event happens with probability at most $0.6^{O(\log(1/\delta))}$, bounding the failure probability to be at most $\delta$.

    \paragraph{Query time:} The query time consist of the time required to compute $M_i x^{(i)}_j$, the time
    required to compute $M_i u_i$, and the time needed to compute $U y$ where $i, j \in \{1,2,3,4\}$.
    The sampling in Step 1 will be done $O(\log(1/\delta) \cdot 1/\alpha^5)$ times, and for each sampled vectors, $\DS_{M_i}$
    is applied $4$ times. Also, the verification in Step 3 consists of computing the inner product for $O(\log(1/\delta))$ many vectors.
    Thus, the total query time is equal to
    \begin{align*}
        O(\log(1/\delta) \cdot 1/\alpha^5) \cdot 4 \cdot (t + O(\log(1/\delta)n) + O(n\log^4(1/\alpha))) \\
    = (4t + n) \cdot \poly(1/\alpha) \cdot \poly\log(1/\delta) \enspace.
    \end{align*}
    This completes the proof of \cref{thm:MV-weak-avg-case}.
\end{proof}

\section{Worst-case to average-case reductions for data structures}

In this section, we show worst-case to average-case reductions in the setting of static data structures. We start by showing a reduction for all linear data structure problems in \cref{sec:all_DS}.

Then, we consider a more powerful type of reductions, which can be used to derive worst-case algorithms from data structures that only satisfy a weak average case condition over both inputs and queries (similarly to the setting of online matrix-vector multiplication). On the negative side, we give a counterexample, showing that general weak-average-case reductions cannot hold for all linear problems. On the positive side, we show that the problem of evaluating a multivariate polynomial admits such a weak-average-case reduction. We stress that as opposed to the online matrix-vector multiplication problem discussed above, the problem of multivariate polynomial evaluation is an example of a non-linear problem admitting such a reduction.

\subsection{Average-case reductions for all linear problems}
\label{sec:all_DS}

Recall that in the setting of data structures, a linear problem over a field $\F$ is defined by a matrix $A \in \F^{m \times n}$. The input to the data structure is a vector $x \in \F^{n}$, which is preprocessed into $s$ memory cells. Then, given queries of the form $i \in [m]$, the goal of the data structure is to output $\ip{A_i,x}$, where $A_i$ is the $i$'th row of $A$. We show a worst-case to average-case reduction for data structures for \emph{all} linear problems. 

\begin{remark}
We note that the presented reduction results in \emph{uniform} data structures. That is, we give an efficient and simple procedure that, given an average-case data structure, creates a worst-case data structure in a black-box way that works for all values of~$n$. 

There is a trivial folklore argument that transforms an average-case data structure into a \text{\emph{non-uniform}} worst-case data structure as follows. Let $X\subseteq\F^n$ of size $|X|\geq\alpha|\F|^n$ be the set where for a given~$n$, the average-case data computes all queries correctly. By the probabilistic method, there exists $(n/\alpha)\log{|\F|}$ shifts of $X$ that cover all of $\F^n$. For every~$n$, a non-uniform data structure will remember all of those shifts, and for each shift $s\in\F^n$, it will also remember the product $As$. Now, given an input vector~$x$, in the preprocessing stage, the data structure just stores the index of a shift $s$ such that $x+s\in X$, and in the query phase it reads the index of the shift and, thus, learns $As$. Now, since $x+s\in X$, we can use the average-case data structure to compute $A(x+s)$, and, finally, compute $Ax=A(x+s)-As$. This results in a non-uniform worst-case data structure whose space complexity and query time differ from those of the average-case data structure by an additive term of $\log((n/\alpha)\log{|\F|})$.
\end{remark}
\DSstrong*
\begin{proof}
    Consider the data structure $\DS$ for the matrix $A$, implied by the assumption of the theorem. There exists a subset $X \seq \F^n$ of size $\abs{X} \geq \alpha \abs{\F}^n$ such that for every input $x \in X$ the data structure answers correctly all queries to the data structure, i.e., $\DS_x(i) = \ip{A_i,x}$ for all $i \in [m]$ and $x \in X$.

    We design a data structure $\DS'$ that outputs correct answers in worst case as follows. Let $x \in \F^n$ be the input to $\DS'$. We start by describing preprocessing and query phases of the data structure.
    
\begin{tcolorbox}[title=Worst-case data structure for $L_A$]    
    \paragraph{Preprocessing:}
        Given $x \in \F^n$ we apply the \cref{cor:y=4x+s} on $x$ with the set $X$, and obtain a non-negative integer $t \leq O(\log^4(1/\alpha))$, a vector $v \in \F^n$ with at most $t$ non-zero entries, and $x_1,x_2,x_3,x_4 \in X$ such that 
            \begin{equation*}
                x = x_1 + x_2 - x_3 - x_4 + v
                \;.
            \end{equation*}
        Furthermore, by \cref{cor:y=4x+s} such decomposition can be found using $\exp(\log^4(1/\alpha)) \cdot \poly\log(1/\delta) \cdot \poly(n)$ field operations with probability at least $1-\delta$.
    
    Then, we use the preprocessing algorithm of the average-case data structure on each one of the $x_j$'s to obtain the algorithms $DS_{x_1}(\cdot), DS_{x_2}(\cdot), DS_{x_3}(\cdot), DS_{x_4}(\cdot)$. Finally, we store the sparse shift vector $v$ by storing the $t$ coordinates, and their values.
    Therefore, the amount of memory used is $4s + O(\log^4(1/\alpha) \log(n))$.
    
    \vspace{10pt}
    \paragraph{Query phase:}
    Given a query $i \in [m]$, we invoke our four instantiations of the average-case data structure stored in the preprocessing stage and  compute $\DS_{x_1}(i), \DS_{x_2}(i), \DS_{x_3}(i), \DS_{x_4}(i)$. We then compute $\ip{A_i,v}$ and return 
    \begin{equation*}
        \DS_{x_1}(i) + \DS_{x_2}(i) - \DS_{x_3}(i) - \DS_{x_4}(i) + \ip{A_i,v} \;.
    \end{equation*}
\end{tcolorbox}

    \paragraph{Complexity:}
    The time and amount of memory used follow immediately from the description. Namely, note that applying the local correction lemma, which dominates the time complexity, is done in time $\exp(\log^4(1/\alpha)) \cdot \poly\log(1/\delta) \cdot \poly(n)$. Hence the total preprocessing time is $4p + \exp(\log^4(1/\alpha)) \cdot \poly(n)$.

    In terms of memory, we store $4$ instances of the average-case data structure $\DS$, where each instance requires $s$ memory cells. In addition we store the sparse vector $v$, by storing its $t$ non-zero indices and their values. Hence the total memory required is $4s+O(\log^4(1/\alpha) \cdot \log(n))$.

    Finally, the bound on the query time consists of $4$ queries to the the average case data structure $\DS$, as well as the computation of $\ip{A_i,v}$. The latter can be done by reading the description of the $t$-sparse vector $v$, and computing their inner product with the corresponding $t$ entries in the $i$'th row of $A$. Hence the total query time is $4t+O(\log^4(1/\alpha) \cdot \log(n))$.
    
    \paragraph{Correctness:}
    By \cref{cor:y=4x+s} we have
    \begin{equation*}
        x = x_1 + x_2 - ( x_3 + x_4) + v
        \;,
    \end{equation*}
    where $x_1,x_2,x_3,x_4 \in X$. By the definition of $X$, this implies that the average-case data structure $\DS$ computes these points correctly, hence
    \begin{equation*}
        \DS_{x_1}(i) + \DS_{x_2}(i) - \DS_{x_3}(i) - \DS_{x_4}(i) =
        \ip{A_i, x_1} + \ip{A_i, x_2} - \ip{A_i, x_3} - \ip{A_i, x_4}  \;,
    \end{equation*}
Furthermore, we directly compute $\ip{A_i,v}$, and hence, by the linearly of the inner product operation, it follows that 
\begin{equation*}
    \ip{A_i,x} = \sum_{j=1}^4\ip{A_i, x_j} + \ip{A_i,v} \;.
\end{equation*}
This concludes the proof of \cref{thm:ds-lin-strong-avg-case}.
\end{proof}

\subsection{Weak-average-case data structures}
In \cref{sec:all_DS}, we showed a worst-case to average-case reduction for all linear problems in the setting of data structures. In the following, we show how to obtain worst-case algorithms starting from a very weak, but natural, notion of average-case reductions that we discuss next.

Recall that in the standard definition of average-case data structures, the algorithm preprocesses its input and is then required to correctly answer all queries for an $\alpha$-fraction of all possible inputs. However, in many cases (such as in the online matrix-vector multiplication problem), we only have an average-case guarantee on both inputs and queries. In this setting, we should first ask what is a natural notion of an average-case condition.

A strong requirement for an average-case algorithm in this case is to correctly answer \emph{all queries} for at least $\alpha$-fraction of the inputs. However, a more desirable condition is to require the algorithm to correctly answer on an \emph{average input and query}. This is captured by the following definition.

\begin{definition}
    A \emph{weak average-case} data structure for computing a function $f\colon \F^n \times Q \to \F^k$ with success rate $\alpha>0$ receives an input $x\in\F^n$, which is preprocessed into $s$ memory cells. Then, given a query $q \in Q$, the data structure $\DS_x(q)$ outputs $y\in\F^{k}$ such that
    \begin{equation*}
        \Pr_{x\in\F^n, q \in Q }[\DS_x(q) = f(x,q)] \geq \alpha \;.
    \end{equation*}
\end{definition}

The challenge in this setting is that the errors may be distributed between both the inputs and the queries. On one extreme, the error could be concentrated on selected inputs, and then the data structure computes \emph{all queries} correctly for $\alpha$-fraction of the inputs. On the other extreme, the error could be spread over all inputs, and then the data structure may only answer $\alpha$-{fraction of the queries} on any inputs. Of course, the error could be distributed anywhere in between these two extremes.

\paragraph{Weak-average-case reductions for matrix-vector multiplication.}
As a first example of the weak-average-case paradigm, we note that our reduction for the online matrix-vector multiplication problem in \cref{sec:OMV} can be cast as a worst-case to weak-average-case reduction. Namely, we start with a data structure that receives a matrix $M\in\F^{n \times n}$ as an input, preprocesses it into $s$ memory cells. Then, on query $v\in \F^n$ the data structure algorithm $\DS_M$ satisfies
\begin{equation*}
    \Pr_{M\in\F^{n \times n}, v\in\F^n}[\DS_M(v) = Mv] \geq \alpha \;.
\end{equation*}
Hence we immediately obtain the following statement.
\MV*

An immediate question is whether it is possible to obtain worst-case to weak-average-case reductions not only for the matrix-vector multiplication problem, but rather for \emph{all} linear problems, as we have in the setting of (standard) average-case data structure. Alas, as we show next, such a general result is impossible.

\subsection{Impossibility of weak-average-case reductions for all linear problems}
\label{sec:impossibility}

We observe that for weak-average-case data structures, there is a simple counterexample which shows that it is \emph{impossible} to obtain worst-case to weak-average-case reductions for \emph{all} linear problems. Nevertheless, we later show that it is possible to obtain such reductions for specific natural problems beyond matrix-vector multiplication, namely for the (non-linear) problem of multivariate polynomial evaluation.

To see the counterexample, first note that a weak-average-case data structure can be equivalently thought of as a data structure where the answer to each input is a vector, rather than a scalar, and the requirement is that the algorithms on average outputs a partially correct vector. That is, a weak-average-case data structure computes a function $f\colon \F^n \to \F^m$ with success rate $\alpha>0$ if after the preprocessing, on query $x \in \F^n$ it satisfies that $\Pr_{x\in\F^n, i \in [m] }[\DS(x)_i = f(x)_i] \geq \alpha$.

\paragraph{Weak-average-case circuits.}
For simplicity, we start with a counterexample for weak-average-case circuits, then extend it to the setting of data structures. Note that the number of linear functions $f\colon \F_2^n \to \F_2^m$ is $2^{nm}$. The total number of (not necessarily linear) circuits with $g$ gates is $2^{O(g \log(g))}$ (see, e.g., Lemma~1.12 in~\cite{J12}). Therefore, by a simple counting argument, a random \emph{linear} function requires a circuit with $g \geq \Omega\left(\frac{mn}{\log(mn)}\right)$ gates.

Now fix a linear function~$f$ of complexity at least $\Omega\left(\frac{mn}{\log(mn)}\right)$, and an arbitrarily small constant $\eps>0$. Let us consider the function $h: \F_2^n \to \F_2^{m/ \eps}$ that embeds $f$ as follows: the first $m$ outputs of $h(x)$ compute $f(x)\in\F^m$, and the remaining $m/\eps-m$ outputs are always zeros. Note that $h$ is also a linear function that requires a circuit with at least $\Omega\left(\frac{mn}{\log(mn)}\right)$ gates.

On the other hand, note that the trivial circuit outputting $m/ \eps$ zeros well-approximates the function $h$, i.e., it satisfies the weak-average-case with success rate $\alpha = (1-\eps)$. Thus, any worst-case to average-case reduction for all functions in this setting would have to blow up the size of the trivial circuit computing $0$ to the size of at least $\Omega\left(\frac{mn}{\log(mn)}\right)$, which is almost the biggest circuit with this given number of inputs and outputs. Therefore, such a reduction would be degenerate.

\paragraph{Weak-average-case data structures.}
Moving on to the setting of data structures, here there is an issue with such an argument. To see that, recall that we want start by picking a \emph{linear} function that is hard even against \emph{non-linear} data structures. While the number of linear functions is still $2^{mn}$, a data structure can compute $s$ \emph{arbitrary} (i.e., not necessarily linear) functions in the preprocessing stage. The problem is that even one such function gives a data structure $2^{2^n}$ possibilities which is already larger than the number of \emph{linear} functions, and so the counting argument here doesn't work.

Nevertheless, we can still get essentially the same result for data structures by the following argument. Let $C(n,m)$ be the complexity of the hardest data structure for a \emph{linear} problem. Then, again, we take the hardest linear function from $n$ bits to $m$ bits, and extend it to a function with $m/\eps$ output bits (where $m/\eps - m$ outputs are constant zeros). The worst-case complexity of this function is at least $C(n, m)$, while the average-case complexity is $0$. Hence, every worst-case to average-case reduction will blow up the size from $0$ to $C(n, m)$. Since every linear function can be computed by a data structure of size $C(n, m/\eps)$ without any reduction, such a reduction is also degenerate.

\subsection{Weak-average-case reductions for multivariate polynomial evaluation}

Our main result in the weak-average-case setting is a worst-case to weak-average-case reduction for data structures computing the (non-linear) problem of multivariate polynomial evaluation. In this problem, the input is a polynomial $q \colon \F^m \to \F$ of total degree at most $d$, given as its coefficients. That is, the length of the input is $n = \binom{m+d}{d}$.
Given the input polynomial, it is preprocessed, and then in the query phase the goal is to respond to each query $x \in \F^m$ with the value $q(x)$. We restate and prove \cref{thm:DS-RM-weak-avg-case} below.

\evalpoly*

\cref{thm:DS-RM-weak-avg-case} follows from the following two lemmas.

\begin{lemma}\label{lemma:DS-RM-0.01-to-0.99}
    Let $\F$ be a prime field, and $d \leq \abs{\F} /10$.
    Suppose that for $\alpha > 2\sqrt{\frac{d}{\abs{\F}}}$ we have
    \begin{equation*}
        \RM_{\F,m,d} \in \DSParams{p}{s}{t}{\Pr_{q, x}[\DS_q(x)= q(x)] \geq \alpha}
        \enspace.
    \end{equation*}
    Then
    \begin{equation*}
        \RM_{\F,m,d} \in \DSParams{4p + \exp(\log^4(1/\alpha)) \cdot \poly\log(1/\delta) \cdot \poly(n)}{4s + O(\log^4(1/\alpha))}{4 \abs{\F}\cdot t + O(\log^4(1/\alpha))) + O(\log(n)}{\text{$\forall q$ with $deg(q) \leq d$} \colon \Pr_{x}[\DS_q(x) = q(x)]  \geq 1-O\left(\sqrt{\frac{d}{\alpha \cdot \abs{\F}}}\right)
        }
        \enspace.
    \end{equation*}

\end{lemma}

\begin{lemma}\label{lemma:DS-RM-0.99-to-worst-case}
   Let $\F$ be a prime field, and $d \leq \abs{\F} /10$.
    Suppose that for $\gamma < 0.1$ we have
    \begin{equation*}
        \RM_{\F,m,d} \in \DSParams{p}{s}{t}{\text{$\forall q$ with $deg(q) \leq d$} \colon \Pr_{x}[\DS_q(x) = q(x)] \geq 1-\gamma }
        \enspace.
    \end{equation*}
    Then
    \begin{equation*}
        \RM_{\F,m,d} \in \DSParams{p}{s}{\abs{\F} \cdot t}{\text{$\forall q$ with $deg(q) \leq d$}, \forall x \in \F^n \colon \Pr[\DS_q(x) = q(x)] \geq 1-4\gamma}
        \enspace.
    \end{equation*}
\end{lemma}

Before proceeding with the proofs of the lemmas above, we will need the following proposition.

\begin{proposition}\label{prop:DS-RM-0.01-to-0.99-for-all-q}
    Let $\F$ be a prime field, $d \leq \abs{\F} /10$, and let $\alpha > 2\sqrt{\frac{d}{\abs{\F}}}$. Let $n = \binom{d+m}{m}$ be the input length---the number of coefficients in a polynomial $q \colon \F^m \to \F$ of total degree at most $d$,

    Let $\DS$ be a data structure for $\RM_{\F,m,d}$ with preprocessing time $p$, that stores $s$ field elements, and has query time $t$.
    Then there exists another data structure $\DS'$ for $\RM_{\F,m,d}$ with preprocessing time $p+n$, that stores $s+m+1$ field elements, has query time $\abs{\F}t$, and satisfies the following guarantee for all input polynomials $q$ of degree at most $d$.
    \begin{equation*}
      \text{If $\Pr_{x}[\DS_q(x)=q(x)] \geq \alpha$, then $\Pr_{x}[\DS'_q(x) = q(x)]  \geq 1-\sqrt{\frac{d}{\alpha \cdot \abs{\F}}}$}
      \enspace.
    \end{equation*}
\end{proposition}

We emphasize that in the claim above the data structures do not depend on the polynomial~$q$.
The proposition says that if $q$ is an input such that $\DS$ outputs the correct evaluation $q(x)$ for at least $\alpha$ fraction of the queries $x$,
then $\DS'$ (which also does not depend on any particular $q$) succeeds on $1-\sqrt{\frac{d}{\alpha \cdot \abs{\F}}}$ fraction of the same input $q$.

\begin{proof}[Proof of \cref{prop:DS-RM-0.01-to-0.99-for-all-q}]
    Denote by $\DS$ the data structure for $\RM_{\F,m,d}$ that outputs the correct answer for at least $\alpha$ fraction of inputs to $q$.
    Below we describe the data structure $\DS'$.

\begin{tcolorbox}[title=Data Structure $\DS'$]
    \paragraph{Preprocessing:} Given the polynomial $q$ of degree at most $d$
    \begin{enumerate}
      \item Run the preprocessing procedure for $\DS$ on the input $q$.
      \item Choose a random \emph{reference point} $\vecw \in \F^m$ and compute $q(\vecw)$.
      \item Store $\vecw$ and $q(\vecw)$ in the memory.
    \end{enumerate}

    \paragraph{Query:} Given a query $\vecx \in \F^m$
    \begin{enumerate}
      \item Consider the line $\ell_{\vecx, \vecw} = \{ \vecx +  r(\vecw - \vecx) : r \in \F\}$ going through $\vecx$ and $\vecw$.
      \item Use the query algorithm of $\DS$ to compute $(\DS_q(z) : z \in \ell_{\vecx,\vecw})$.
      \item Let $Q = \{q_1,q_2,\dots,q_k\}$ be all the univariate polynomials of degree at most $d$
      that agree with $(\DS_q(z) : z \in \ell_{\vecx,\vecw})$ on at least $\alpha/2$ fraction of points in $\ell_{\vecx,\vecw}$.
      (It is possible that $Q = \emptyset$.)
      \item Use the value $q(\vecw)$ from the preprocessing phase, and let $Q'  =\{q' \in Q : q'(\vecw) = q(\vecw)\}$.
      (It is possible that $Q' = \emptyset$.)
      \item Choose $q' \in Q'$ arbitrarily and output $q'(\vecx)$.
    \end{enumerate}
\end{tcolorbox}

    Next, we claim that if $\alpha>2\sqrt{\frac{d}{\abs{\F}}}$, then
    for at least $1-O\left(\sqrt{\frac{d}{\abs{\F \alpha}}}\right)$ fraction of the queries $\vecx$
    the query phase correctly outputs $q(\vecx)$.

    It will be convenient to consider a function $A \colon \F^m \to \F$ defined as $A(z) = \DS_q(z)$ for all $z\in \F^n$.
    Note that $A$ agrees with $q$ on at least $\alpha$-fraction of points.
    Furthermore, note that for simplicity we may assume that $q$ is the all zeros polynomial.
    Indeed, we can define $A'(x) := A(x) - q(x)$, and consider the case where the input is the all zeros
    polynomial, and the query algorithm is $A'$.
    Therefore, (1) we have $\Pr_{x \in \F^m}[A(x) = 0] \geq \alpha$,
    and (2) in the preprocessing phase we know that $q(\vecw) = 0$ for a random point $\vecw$,
    though it is not necessarily true that $A(\vecw) = 0$.

    The following three claims complete the proof of \cref{prop:DS-RM-0.01-to-0.99-for-all-q}.
    \begin{claim}\label{claim:most-ws-are-good1}
    In the preprocessing phase, for a random choice of the reference point $\vecw$
    with high probability over $\vecx$ it holds that $\vec{0} \in Q$, and hence in $Q'$. More formally, we have
    \begin{equation*}
        \E_{\vecw \in \F^m}[\Pr_{x}[\vec{0} \in Q]] \geq 1 - \frac{4}{\abs{\F}\alpha} \enspace.
    \end{equation*}
    In particular, by Markov's inequality, for at least $1 - \sqrt{\frac{4}{\abs{\F}\alpha}}$ of $\vecw$'s it holds that
    \begin{equation}\label{eq:zero-in-Q}
        \Pr_{x}[\vec{0} \in Q] \geq 1 - \sqrt{\frac{4}{\abs{\F}\alpha}} \enspace.
    \end{equation}
    \end{claim}

    \begin{proof}
        It is a standard fact in the literature on derandomization (see, e.g., \cite[Corollary~1.2]{MoshkovitzRaz08})
        that for any set $O \seq \F^n$ of size $\abs{O} = \alpha \abs{\F}^n$,
        and a random line $\ell_{\vecx, \vecw} = \{ \vecx + t( \vecw - \vecx) : t \in \F\}$
        that passes through uniformly random $\vecx, \vecw \in \F^n$ it holds that
        \begin{equation*}
            \Pr\left[\abs{\frac{\abs{\ell_{\vecx, \vecw} \cap O}}{\abs{\ell_{\vecx, \vecw}}} - \alpha} > \eps \right] \leq \frac{1}{\abs{\F}}\frac{\alpha}{\eps^2}
            \enspace.
        \end{equation*}
        The conclusion of the claim follows by letting $O = \{x \in \F^n : A(x) = 0\}$, and $\eps = \alpha/2$.
    \end{proof}

    \begin{claim}[{\cite[Proposition~3.5]{MoshkovitzRaz08}}]\label{claim:list-is-short}
    Choose $\vecw$ and $\vecx$ uniformly at random and consider the line~$\ell_{\vecx,\vecw}$.
    Let $Q = \{q_1,q_2,\dots,q_k\}$ be all the univariate polynomials of degree at most $d$
    that agree with~$A$ on at least $\alpha/2$ fraction of points in $\ell_{\vecx,\vecw}$.
    If $\alpha > 2\sqrt{\frac{d}{\abs{\F}}}$, then $k \leq 2/\alpha$.
    \end{claim}

    \begin{claim}\label{claim:schwartz-zippel}
    Choose $\vecw$ and $\vecx$ uniformly at random and consider the line $\ell_{\vecx,\vecw}$.
    Let $Q = \{q_1,q_2,\dots,q_k\}$ be all the univariate polynomials of degree at most $d$
    that agree with $A$ on at least $\alpha/2$ fraction of points in $\ell_{\vecx,\vecw}$.
    Then, for all $q_i \in Q$ that are not identically zero it holds that
    $\Pr[q_i(\vecw) = 0] \leq \frac{d}{\abs{\F}}$.

    In particular, if $\alpha > 2\sqrt{\frac{d}{\abs{\F}}}$ then
    \begin{equation*}
      \Pr_{\vecw}\left[\Pr_{\vecx}[\forall q_i \in Q \setminus \{\vec{0}\}: q_i(\vecw) \neq 0] \geq 1 - \sqrt{\frac{2d}{\alpha\abs{\F}}} \right] \geq 1 - \sqrt{\frac{2d}{\alpha\abs{\F}}} \enspace.
    \end{equation*}
    \end{claim}
    \begin{proof}
        For any choice of $\vecx$ if $\vecw$ is chosen uniformly at random, and each univariate polynomial $q_i$ is of degree $d$, the by Schwarz-Zippel lemma
        \begin{equation*}
            \Pr[q_i(\vecw) = 0] \leq \frac{d}{\abs{\F}} \enspace.
        \end{equation*}
        Also, by \cref{claim:list-is-short}, we know that $\abs{Q} = k \leq 2/\alpha$, and hence, by union bound
        \begin{equation*}
            \Pr_{\vecx,\vecw}\left[\exists q_i \in Q \setminus \{\vec{0}\}: q_i(\vecw) = 0 \right]
            \leq \frac{k d}{\abs{\F}} \leq \frac{2d}{\alpha\abs{\F}} \enspace.
        \end{equation*}
        This implies 
        \begin{equation*}
            \E_{\vecw}\left[\Pr_{\vecx}[\forall q_i \in Q \setminus \{\vec{0}\}: q_i(\vecw) \neq 0]\right]
            \geq 1 - \frac{k d}{\abs{\F}} \geq 1 - \frac{2d}{\alpha\abs{\F}} \enspace.
        \end{equation*}
        The claim follows by Markov's inequality.
    \end{proof}
    We now return to the proof of \cref{prop:DS-RM-0.01-to-0.99-for-all-q}.
    By the claims above, for most $\vecw$'s it holds that if we choose $\vecw$ as a reference point,
    then for most $\vecx$'s, we have $\vec{0} \in Q$, and there is no other univariate polynomial $q_i$ in $Q$ such that $q_i(\vecw)=0$.
    More precisely, 
    by combining \cref{claim:most-ws-are-good1} with \cref{claim:schwartz-zippel}
    for at least $1 - O\left(\sqrt{\frac{d}{\alpha\abs{\F}}} \right)$ fraction of $\vecw$'s it holds that
    \begin{equation}\label{eq:non-zero-not-in-Q}
      \Pr_x[\vec{0} \in Q \wedge \forall q \in Q : q \not \equiv 0, q(\vecw) \neq 0]  \geq 1 - O\left(\sqrt{\frac{d}{\alpha\abs{\F}}}\right)
      \enspace,
    \end{equation}
    as required.
\end{proof}

Now we proceed with proving \cref{lemma:DS-RM-0.01-to-0.99}.

\begin{proof}[Proof of \cref{lemma:DS-RM-0.01-to-0.99}]
    Suppose there is a data structure $\DS$ as in the assumption of \cref{lemma:DS-RM-0.01-to-0.99},
    with success probability $\Pr_{q, x}[\DS_q(x)= q(x)] \geq \alpha$.
    We show below how to construct a data structure $\DS'$ that will work for all input polynomial $q$ and for most queries $x$.

\begin{tcolorbox}[title=Preprocessing:]
\paragraph{Input:} An input polynomial $q$ of degree at most $d$
      \begin{enumerate}
        \item Identify $q$ with a vector $q \in \F^n$ of its coefficients for $n = \binom{m+d}{d}$.

        \item Let $Z = \{q \in \F^n : \abs{X_q} \geq \frac{\alpha}{2} \cdot \abs{\F}^n \}$.
        
        \item Let $O_Z$ be a membership oracle for $Z$ that given a polynomial $q'$ estimates $\frac{\abs{X_{q'}}}{\abs{\F}^n}$,
        the fraction of points on which $\DS_{q'}$ outputs $q'(x)$ correctly,
        within an additive error of $\alpha/10$, and returns ACCEPT if and only if the estimated fraction is more than $\alpha/3$.\footnote{
        This is done by sampling $O(1/\alpha^2)$ uniformly random $x$'s in $\F^n$, computing $\DS_{q'}(x)$ and $q'(x)$, and comparing the two results.
        In particular, if $\abs{X_{q'}} \geq \frac{\alpha}{2}\cdot \abs{\F}^n$, then $O_Z(q') = ACCEPT$ with probability $1 - \eps$,
        and if $\abs{X_{q'}} \leq \frac{\alpha}{4}\cdot \abs{\F}^n$, then $O_Z(q') = REJECT$ with probability $1 - \eps$.}

        \item By applying \cref{cor:y=4x+s}, with probability $1-\delta$ we obtain a vector $u$ with at most $O(\log^4(1/\alpha))$ non-zero elements such that
        \begin{equation*}
            \Pr_{q_1,q_2,q_3 \in \F^n}[q_1,q_2,-q_3,-q_4 \in Z] \geq \Omega(\alpha^5) \:,
        \end{equation*}
        where 
        $q_4 \in \F^n$ is such that $q-u=q_1+q_2-q_3-q_4$.
        
        \item Therefore, given $q$ and $s$ we can sample $O(\log(1/\delta)\cdot \log^4(1/\alpha))$ triplets of vectors
        until we find a triplet $(q_1,q_2,q_3)$ and let $q_4 = q_1+q_2-q_3-q-u$ satisfying $q_1,q_2 \in Z, -q_3,- q_4 \in Z$.
        Note that we can use the membership oracle $O_Z$ to check that the vectors belong to $Z$.
        
        \item Note that since each $v_i$ belongs to $Z$, we have that $\DS$
        outputs $p_i(x)$ correctly on at least $\alpha/4$ fraction of inputs.
        Thus, we can apply \cref{prop:DS-RM-0.01-to-0.99-for-all-q} on $\DS$ and
        obtain the data structure $\DS'$ such that $\Pr_{x \in \F^n}[\DS'_{q_i}(x) = q_i(x)] \geq 1-O\left(\sqrt{\frac{d}{\alpha \cdot \abs{\F}}}\right)$
        for all $i=1,2$, and  $\Pr_{x \in \F^n}[\DS'_{-q_j}(x) = -q_j(x)] \geq 1-O\left(\sqrt{\frac{d}{\alpha \cdot \abs{\F}}}\right)$ for $j=3,4$.
        
        \item We store all the memory obtained by preprocessing the polynomials $q_1,q_2,q_3,q_4$ with $\DS$.
        We also store the sparse vector $s$ by storing the $O(\log^4(1/\alpha))$ non-zero coordinates, and their values.
      \end{enumerate}
\end{tcolorbox}

For every polynomial $q$ of degree at most $d$, let $X_q = \{x \in \F^m : \DS_q(x) = q(x)\}$.
    By averaging, there is a set $Z$ of degree $d$ polynomials such that
    $\abs{Z} \geq \alpha/2 \cdot \abs{\F}^n$, and $\abs{X_q} \geq \alpha/2 \cdot \abs{\F}^m$ for every $q \in Z$.
    Furthermore, note that it is straightforward to construct a membership oracle $O_Z$ for $Z$,
    that given a polynomial $q$ and access to $\DS_q$ estimates the fraction of queries $x$ on which $\DS_q(x) = q(x)$.
    
Below we describe the preprocessing phase and the query phase of $\DS'$.

    The preprocessing on an input $q$ works as follows.
    By identifying $q$ with the vector of its coefficients in $\F^n$ with $n = \binom{m + d}{m}$,
    we use \cref{cor:y=4x+s} to represent the vector $q$ as $q = q_1 + q_2 - q_3 - q_4 + u$, where each $q_i \in Z$ and $u$ is a sparse vector.
    Then, for each $q_i$ we the reduction from \cref{prop:DS-RM-0.01-to-0.99-for-all-q}
    to obtain a data structure that works for each of the $q_i$ for almost all queries $x$.
    
    In the query phase, we use the data structures for each $q_i$ to compute $q_i(x)$, 
    and for the sparse polynomial $s$, we simply compute $u(x)$ using brute force.
    Finally, we return $q_1(x) + q_2(x) - q_3(x) - q_4(x) + u(x)$.

    \paragraph{Preprocessing time and space:} In the preprocessing step, we store the memory of the preprocessing for each of $q_i$ and 
    in addition the sparse vector $u$.
    Hence, the total space used is $4s + \log^4(1/\alpha)$ field elements + additional $\log^4(1/\alpha)$ coordinates of the input.
    Also, the running time is determined by number of samples needed to construct the oracle in Step 1 using
    \cref{cor:y=4x+s}. Both these steps
    take at most $O(\log^4(1/\alpha) \cdot \log(1/\delta))$ samples, bounding the running time of the preprocessing step.

    Next we describe the query phase of our data structure.

\begin{tcolorbox}[title=Query phase:]
    \paragraph{Input:} A query $\vecx \in \F^m$. 
    
    Recall that for the polynomial $q$, we
    have stored high-agreement data structures for evaluating $q_1, q_2, q_3, q_3$, together with a polynomial which is represented by a sparse vector of coefficient $u$.

    \begin{enumerate}
        \item
         For each polynomial $q_i$ let $y_i = \DS'_{q_i}(\vecx)$.
      \item
        Compute $u(x)$. Since $u$ has at most $O(\log^4(1/\alpha))$ non-zero coordinates, it follows that $u(x)$ can be computed in query time $O(\log^4(1/\alpha) \log(n))$.
      \item
        Return $y_1 + y_2 - y_3 - y_4 + u(x)$.
    \end{enumerate}
\end{tcolorbox}

    \paragraph{Query time:} The query time consists of querying the data structure $4$ times, and evaluating $u(x)$.
    Note that the high-agreement data structure makes $\abs{\F}$ queries
    to the weak-average-case data structure, and each query takes time $t$. Thus, the total
    query time is $4\abs{\F}\cdot t + O(\log^4(1/\alpha) \log(n)))$.

    \paragraph{Correctness:} To prove the correctness, we bound the failure probability of the algorithm.
    Note that the algorithm returns the correct answer, unless for one of the polynomials $q_i$ it holds that,
    $\DS'_{q_i}(\vecx) \neq q_i(\vecx)$. This event happens with probability at most
    $O(\sqrt{\frac{d}{\abs{\F}\alpha}})$.
    Hence, by applying union bound we can bound the failure probability to $O(\sqrt{\frac{d}{\abs{\F}\alpha}})$.
\end{proof}

Finally, we prove \cref{lemma:DS-RM-0.99-to-worst-case}.
\begin{proof}[Proof of \cref{lemma:DS-RM-0.99-to-worst-case}]
    The proof of this lemma basically relies on the local decoding algorithm for Reed-Muller codes.
    Given a point $\vecx$, the query algorithm samples a random line $\ell_{x,y} = \{ \vecx + r (\vecy - \vecx) : r \in \F\}$ passing through $\vecx$, and queries the data structure for all the points on this line.
    Given these values, the algorithm finds the closest univariate polynomial of degree at most $d$, call it $h$,
    and outputs $h(\vecx)$. It is not hard to see that the algorithm succeeds with probability at least $1-4\gamma$.

    For correctness since $\DS_p$ agrees with $q$ on $1-\gamma$ fraction of the points $z \in \F^n$,
    it follows that for a random line $\ell$ through $x$
    the data structure $\DS$ satisfies
    $\Pr[agr \geq 3/4] \geq 1-4\gamma$,
    where $agr$ denotes the fraction of points $z \in \ell$ with $\DS_{q_i}(z) = q_i(z)$.
    For each such line $\ell$ the only polynomial that agrees with $\DS$ on $\ell$ is $q_i$,
    and hence with probability at least $1-4\gamma$ the data structure outputs $q_i(x)$, as required.
\end{proof}

\paragraph{Putting it all together:} Below we summarize the reductions above, and describe the full reduction
that given a weak-average-case data structure $\DS$ that computes the correct answer for only $\alpha$ fraction of $(q,x)$,
gives us a data structure that works with high probability for all inputs $q$ and all queries $x$.

We first use \cref{lemma:DS-RM-0.99-to-worst-case}, reducing the problem to evaluating $p$ on a random line $\ell$ passing through $\vec{x}$.
Then, we apply \cref{lemma:DS-RM-0.01-to-0.99}, to write $p$ as sum of $5$ polynomials, for which we know one of them is sparse and can be computed efficiently,
and the other four belong to the set of \textit{good} polynomials, i.e., those polynomials for which the data structure succeeds in evaluating them on all
but a small fraction of inputs. Finally, for each of these polynomials we apply the reduction in \cref{prop:DS-RM-0.01-to-0.99-for-all-q}.
This last step corresponds to choosing a random reference point $\vecw \in \F^n$,
and passing lines between every $\vecz \in \ell$ and $\vecw$, and evaluating each of the $q_i$ on each of the $\F$ lines.

\bibliographystyle{alpha}
\bibliography{refs}

\begin{thebibliography}{BRTW14}

\bibitem[AD97]{AD97}
Mikl{\'o}s Ajtai and Cynthia Dwork.
\newblock A public-key cryptosystem with worst-case/average-case equivalence.
\newblock In {\em STOC 1997}, pages 284--293, 1997.

\bibitem[Ajt96]{A1996}
Mikl{\'o}s Ajtai.
\newblock Generating hard instances of lattice problems.
\newblock In {\em STOC 1996}, pages 99--108, 1996.

\bibitem[Aka08]{Akavia08}
Adi Akavia.
\newblock Learning significant fourier coefficients over finite abelian groups.
\newblock In Ming{-}Yang Kao, editor, {\em Encyclopedia of Algorithms - 2008
  Edition}. Springer, 2008.

\bibitem[AMN98]{azar1998approximating}
Yossi Azar, Rajeev Motwani, and Joseph~Seffi Naor.
\newblock Approximating probability distributions using small sample spaces.
\newblock {\em Combinatorica}, 18(2):151--171, 1998.

\bibitem[AV21]{AV21}
Josh Alman and Virginia {Vassilevska Williams}.
\newblock A refined laser method and faster matrix multiplication.
\newblock In {\em SODA 2021}, pages 522--539. SIAM, 2021.

\bibitem[BABB19]{BBB21}
Enric Boix-Adser{\`a}, Matthew Brennan, and Guy Bresler.
\newblock The average-case complexity of counting cliques in
  {Erd\H{o}s--R\'enyi} hypergraphs.
\newblock In {\em FOCS 2019}, pages 1256--1280. IEEE, 2019.

\bibitem[BFNW93]{BFNW93}
L{\'{a}}szl{\'{o}} Babai, Lance Fortnow, Noam Nisan, and Avi Wigderson.
\newblock {BPP} has subexponential time simulations unless {EXPTIME} has
  publishable proofs.
\newblock {\em Computational Complexity}, 3(4):307--318, 1993.

\bibitem[BLR90]{BlumLR90}
Manuel Blum, Michael Luby, and Ronitt Rubinfeld.
\newblock Self-testing/correcting with applications to numerical problems.
\newblock In {\em STOC 1990}, pages 73--83. {ACM}, 1990.

\bibitem[BRSV17]{BRSV17}
Marshall Ball, Alon Rosen, Manuel Sabin, and Prashant~Nalini Vasudevan.
\newblock Average-case fine-grained hardness.
\newblock In {\em STOC 2017}, pages 483--496, 2017.

\bibitem[BRSV18]{BRSV18}
Marshall Ball, Alon Rosen, Manuel Sabin, and Prashant~Nalini Vasudevan.
\newblock Proofs of work from worst-case assumptions.
\newblock In {\em CRYPTO 2018}, pages 789--819. Springer, 2018.

\bibitem[BRTW14]{BenSassonRZTW14}
Eli Ben{-}Sasson, Noga Ron{-}Zewi, Madhur Tulsiani, and Julia Wolf.
\newblock Sampling-based proofs of almost-periodicity results and algorithmic
  applications.
\newblock In {\em {ICALP} 2014}, pages 955--966. Springer, 2014.

\bibitem[BT06]{BT06}
Andrej Bogdanov and Luca Trevisan.
\newblock Average-case complexity.
\newblock {\em Foundations and Trends in Theoretical Computer Science},
  2(1):1--106, 2006.

\bibitem[CGL15]{CGL15}
Raphael Clifford, Allan Gr{\o}nlund, and Kasper~Green Larsen.
\newblock New unconditional hardness results for dynamic and online problems.
\newblock In {\em FOCS 2015}, pages 1089--1107. IEEE, 2015.

\bibitem[Cha02]{Chang02}
Mei-Chu Chang.
\newblock {A polynomial bound in Freiman's theorem}.
\newblock {\em Duke Mathematical Journal}, 113(3):399 -- 419, 2002.

\bibitem[CKL18]{CKL18}
Diptarka Chakraborty, Lior Kamma, and Kasper~Green Larsen.
\newblock Tight cell probe bounds for succinct boolean matrix-vector
  multiplication.
\newblock In {\em STOC 2018}, pages 1297--1306, 2018.

\bibitem[CKLM18]{CKLM18}
Arkadev Chattopadhyay, Michal Kouck{\`y}, Bruno Loff, and Sagnik Mukhopadhyay.
\newblock Simulation beats richness: New data-structure lower bounds.
\newblock In {\em STOC 2018}, pages 1013--1020, 2018.

\bibitem[CS10]{CS10}
Ernie Croot and Olof Sisask.
\newblock A probabilistic technique for finding almost-periods of convolutions.
\newblock {\em Geometric and Functional Analysis}, 20(6):1367–--1396, 2010.

\bibitem[DKKS21]{DKKS21}
Pavel Dvo{\v{r}}{\'a}k, Michal Kouck{\`y}, Karel Kr{\'a}l, and Veronika
  Sl{\'\i}vov{\'a}.
\newblock Data structures lower bounds and popular conjectures.
\newblock {\em arXiv:2102.09294}, 2021.

\bibitem[DLV20]{DLV20}
Mina Dalirrooyfard, Andrea Lincoln, and Virginia {Vassilevska Williams}.
\newblock New techniques for proving fine-grained average-case hardness.
\newblock In {\em FOCS 2020}, pages 774--785. IEEE, 2020.

\bibitem[FF93]{FF93}
Joan Feigenbaum and Lance Fortnow.
\newblock Random-self-reducibility of complete sets.
\newblock {\em SIAM Journal on Computing}, 22(5):994--1005, 1993.

\bibitem[FHM01]{FHM01}
Gudmund~Skovbjerg Frandsen, Johan~P. Hansen, and Peter~Bro Miltersen.
\newblock Lower bounds for dynamic algebraic problems.
\newblock {\em Information and Computation}, 171(2):333--349, 2001.

\bibitem[Fre77]{freivalds1977probabilistic}
Rusins Freivalds.
\newblock Probabilistic machines can use less running time.
\newblock In {\em IFIP 1977}, pages 839--842, 1977.

\bibitem[GR18]{GR18}
Oded Goldreich and Guy Rothblum.
\newblock Counting $t$-cliques: Worst-case to average-case reductions and
  direct interactive proof systems.
\newblock In {\em FOCS 2018}, pages 77--88. IEEE, 2018.

\bibitem[HKNS15]{HKNS15}
Monika Henzinger, Sebastian Krinninger, Danupon Nanongkai, and Thatchaphol
  Saranurak.
\newblock Unifying and strengthening hardness for dynamic problems via the
  online matrix-vector multiplication conjecture.
\newblock In {\em STOC 2015}, pages 21--30, 2015.

\bibitem[HLS21]{HLS21}
Monika Henzinger, Andrea Lincoln, and Barna Saha.
\newblock The complexity of average-case dynamic subgraph counting.
\newblock {\em {ECCC}}, 2021.

\bibitem[Imp95]{I95}
Russell Impagliazzo.
\newblock A personal view of average-case complexity.
\newblock In {\em CCC 1995}, pages 134--147. IEEE, 1995.

\bibitem[Imp11]{I11}
Russell Impagliazzo.
\newblock Relativized separations of worst-case and average-case complexities
  for {NP}.
\newblock In {\em CCC 2011}, pages 104--114. IEEE, 2011.

\bibitem[Juk12]{J12}
Stasys Jukna.
\newblock {\em Boolean function complexity: advances and frontiers}, volume~27.
\newblock Springer Science \& Business Media, 2012.

\bibitem[KU08]{KU08}
Kiran~S. Kedlaya and Christopher Umans.
\newblock Fast modular composition in any characteristic.
\newblock In {\em FOCS 2008}, pages 146--155. IEEE, 2008.

\bibitem[Lar12]{L12}
Kasper~Green Larsen.
\newblock Higher cell probe lower bounds for evaluating polynomials.
\newblock In {\em FOCS 2012}, pages 293--301. IEEE, 2012.

\bibitem[Lev86]{L86}
Leonid~A. Levin.
\newblock Average case complete problems.
\newblock {\em SIAM Journal on Computing}, 15(1):285--286, 1986.

\bibitem[Lip91]{L91}
Richard Lipton.
\newblock New directions in testing.
\newblock {\em Distributed computing and cryptography}, 2:191--202, 1991.

\bibitem[LLV19]{LLW19}
Rio LaVigne, Andrea Lincoln, and Virginia {Vassilevska Williams}.
\newblock Public-key cryptography in the fine-grained setting.
\newblock In {\em CRYPTO 2019}, pages 605--635. Springer, 2019.

\bibitem[Lov15]{Lovett15}
Shachar Lovett.
\newblock An exposition of {Sanders'} quasi-polynomial {Freiman-Ruzsa} theorem.
\newblock {\em Theory of Computing}, pages 1--14, 2015.

\bibitem[Lov17]{Lovett17}
Shachar Lovett.
\newblock Additive combinatorics and its applications in theoretical computer
  science.
\newblock {\em Theory of Computing}, pages 1--55, 2017.

\bibitem[LW17]{LW17}
Kasper~Green Larsen and Ryan Williams.
\newblock Faster online matrix-vector multiplication.
\newblock In {\em SODA 2017}, pages 2182--2189. SIAM, 2017.

\bibitem[MR06]{MoshkovitzRaz08}
Dana Moshkovitz and Ran Raz.
\newblock Sub-constant error low degree test of almost-linear size.
\newblock In {\em STOC 2006}, pages 21--30. ACM, 2006.

\bibitem[NN93]{NN93}
Joseph Naor and Moni Naor.
\newblock Small-bias probability spaces: Efficient constructions and
  applications.
\newblock {\em SIAM J. Comput.}, 22(4):838--856, 1993.

\bibitem[Reg04]{R04}
Oded Regev.
\newblock New lattice-based cryptographic constructions.
\newblock {\em Journal of the ACM}, 51(6):899--942, 2004.

\bibitem[San12]{Sanders12}
Tom Sanders.
\newblock On the {Bogolyubov–Ruzsa} lemma.
\newblock {\em IEEE Trans. Inf. Theory}, 5(3):627--655, 2012.

\bibitem[Sho09]{S09}
Victor Shoup.
\newblock {\em A computational introduction to number theory and algebra}.
\newblock Cambridge, 2009.

\bibitem[STV01]{STV01}
Madhu Sudan, Luca Trevisan, and Salil Vadhan.
\newblock Pseudorandom generators without the {XOR} lemma.
\newblock {\em Journal of Computer and System Sciences}, 62(2):236--266, 2001.

\bibitem[TS17]{TaShma17}
Amnon Ta-Shma.
\newblock Explicit, almost optimal, epsilon-balanced codes.
\newblock In {\em STOC 2017}, pages 238--251, 2017.

\bibitem[{Vas}18]{V18}
Virginia {Vassilevska Williams}.
\newblock On some fine-grained questions in algorithms and complexity.
\newblock In {\em ICM 2018}, 2018.

\end{thebibliography}
\appendix

\newpage
\section{Proof of the probabilistic version of Sanders' lemma}\label{sec:appendix-sanders}

Below we prove \cref{thm:bogolyubov-quasipoly}. The proof follows the approach of Sanders, with several modifications. We follow the exposition of Lovett~\cite{Lovett15} in the proof of the lemma.

\sandersThm*

Before starting with the proof, let us establish some notation.
For a set $A \seq \F^n$, we denote by $1_A \colon \F^n \to \{0,1\}$ the indicator function of $A$,
where $1_A(x) = 1$ if $x \in A$, and $1_A(x) = 0$ otherwise.
In particular $\E_{x \in \F^n}[1_A(x)] = \frac{\abs{A}}{\abs{\F}^n}$.
We also let $\varphi_A \colon \F^n \to \R$ be the normalization of $1_A$ defined as $\varphi_A(x) = 1_A(x) \cdot \frac{\abs{\F}^n}{\abs{A}}$ so that $\E_{x \in \F^n}[\varphi_A(x)] = 1$.
When the set $A$ is a singleton $A = \{a\}$, we will write $\varphi_a = \varphi_{\{a\}}$.

It is easy to verify that for a set $A$ and a function $f$ the convolution $\varphi_A*f$ is given by
$\varphi_A*f(x) = \E_{a \in A}[f(x-a)]$.
In particular, for $a \in \F^n$ we have $\varphi_a*f(x) = f(x-a)$.

As a starting point, we define the set $D = \{d \in \F^n : 1_A * 1_{-A}(d) \geq \delta\}$ for a parameter $\delta = \alpha^2/20$.
That is, $D$ is the set of all \emph{popular differences} of two elements of $A$. In other words, $D$ consists of all $d \in \F^n$ such that
there are $\delta\abs{\F}^n$ pairs $(a,a') \in A^2$ satisfying $d=a-a'$, i.e., $\Pr_{a \in \F^n, a' = d-a}[a \in A, a' \in -A] \geq \delta$.

Note first that $\ip{1_{A-A}, \varphi_A * \varphi_{-A}} = \E_{x,y \in \F^n}[1_{A-A}(x-y)\varphi_A(x)\varphi_{-A}(-y)] = \E_{x,y \in A}[1_{A-A}(x-y)] = 1$.
Next we observe that $D$ approximates $A-A$, in the sense that $\ip{1_{D}, \varphi_A * \varphi_{-A}} \geq 1 - \frac{\delta}{\alpha^2} = 0.95$.
Indeed, using the fact that $D \seq A-A$, we have
\begin{align}
    \ip{1_D, \varphi_A * \varphi_{-A}}
    &= 
    \ip{1_{A-A}, \varphi_A * \varphi_{-A}} - 
    \ip{1_{\F^n \setminus D}, \varphi_A * \varphi_{-A}}
    \nonumber\\
    &= 
    1 - 
    \frac{1}{\alpha^2}\ip{1_{\F^n \setminus D}, 1_A * 1_{-A}}
    \nonumber\\
    &= 1 - \frac{1}{\alpha^2}\cdot \Pr_{d,a \in \F^n} [a \in A, d-a \in -A | d \notin D] \cdot \Pr_{d \in \F^n}[d \notin D]
         \nonumber\\
     &\geq 1 - \frac{\delta}{\alpha^2}\,.\label{eq:D-approx-A-A}
 \end{align}

We remark that this is one of the main differences in our proof compared to the original proof of Sanders,
who only relied on the fact that $\ip{1_{A-A}, \varphi_A * \varphi_{-A}} = 1$.

The proof of \cref{thm:bogolyubov-quasipoly} consists of the following two parts.

\begin{lemma}\label{lemma:sanders-step1}
    Let $A \seq \F^n$ be a set of size $\abs{A} = \alpha \abs{\F}^n$. Set $t = O(\log(1/\alpha))$.
    There exists a set $X \seq \F^n$ of size $\abs{X} \geq \alpha^{O(\log^3(1/\alpha))} \abs{\F}^n$ such that for all $x_1,x_2,\dots,x_t \in X$
    it holds that
    \begin{equation}\label{eq:sanders-step1}
        \Pr_{a_1,a_2 \in A}[a_1 - a_2 - \sum_{i=1}^t x_i \in D] \geq 0.9
        \enspace.
    \end{equation}
\end{lemma}

Given the set $X$ from \cref{lemma:sanders-step1} we use a standard Fourier-analytic argument
to define a large subspace $V$ such that $\abs{D \cap V'} \geq 0.8\abs{V}$ where
$V'$ is some coset of $V$. In fact, we show that there are many such cosets.
Formally, we prove the following lemma.

\begin{lemma}\label{lemma:sanders-step2}
    Let $A \seq \F^n$ be a set of size $\abs{A} = \alpha \abs{\F}^n$,
    Then, there exists a subspace $V \seq \F^n$ of dimension $\dim(V) = n - O(\log^4(1/\alpha))$
    and a vector $b \in \F^n$ such that if we sample a uniformly random $a \in A$ and $v \in V$ then
    \begin{equation}\label{eq:sanders-step2}
        \Pr_{a \in A, v \in V} [v+a+b \in D] \geq 0.85
        \enspace.
    \end{equation}
\end{lemma}

\begin{remark}
    \cref{lemma:sanders-step1} corresponds to Lemma~4.2 in \cite{Lovett15}.
    The only difference is that we claim that the sum belongs to $D$ with high probability,
    while in \cite{Lovett15} the sum belongs to $A-A$. Note that the two statements are indeed
    close to each other by \cref{eq:D-approx-A-A}.
    
    \cref{lemma:sanders-step2} roughly corresponds to the conclusion of the section \emph{``A Fourier-analytic argument''} in \cite{Lovett15}.
\end{remark}

We show next how to conclude the proof of \cref{thm:bogolyubov-quasipoly} from \cref{lemma:sanders-step2}.
Indeed, by \cref{eq:sanders-step2} it follows that
$\E_{a \in A}[\Pr_{v \in V}[v+a+b \in D]] \geq 0.85$,
and hence for at least $0.05 \abs{A}$ many $a \in A$ it holds that
$\Pr_{v \in V}[v+a+b \in D] \geq 0.8$.
In particular, for  $C = \{ c \in \F^n: \abs{D \cap (V+c)} \geq 0.8\abs{V} \}$
we have $\abs{C} \geq 0.05\alpha \abs{\F}^n$.

\begin{claim}
    The set $C \seq \F^n$ is symmetric. Namely, if $c \in C$, then $-c \in C$.
\end{claim}
\begin{proof}
    Let $c \in C$ and let $V+c$ be the corresponding coset of $V$.
    We claim that $-c \in C$, i.e.,
    $\abs{\{v \in V : \text{    there are $\geq \delta \abs{\F}^n$ pairs $(a_1,a_2) \in A^2$ such that $a_1-a_2 = v-c$}\}} \geq 0.8\abs{V}$.
    
    For $c \in C$ let $P_c = \{v \in V : \text{there are $\geq \delta \abs{\F}^n$ pairs $(a_1,a_2) \in A^2$ such that $a_1-a_2 = -v+c$}\}$. By definition of $C$, we have $\abs{P_c} \geq 0.8\abs{V}$.

    To see that $-c \in C$ take any $v \in P_C$, and note that $a_1-a_2 = -v+c$ if and only if $a_2-a_1 = v-c$.
    Therefore, for each $v \in P_c$ there are
    $\geq \delta \abs{\F}^n$ pairs $(a_1,a_2) \in A^2$ such that $a_2-a_1 = v-c$, and thus $-c \in C$, as required.
\end{proof}

Let us choose a unique representative $c^*$ for each coset $V+c$ of $V$ such that $\abs{D \cap (V+c)} \geq 0.8\abs{V}$,
and let $C^* = \{ \mbox{$c^*$ is the representative of $V+c$}: \abs{D \cap (V+c)} \geq 0.8\abs{V} \}$.
And furthermore, let us assume without loss of generality that $C^*$ is symmetric, i.e. $c^* \in C^*$ implies that $-c^* \in C^*$.
Then, the union of all these coset covers is at least $0.05\alpha$ fraction of $\F^n$,
i.e.,
\begin{equation}\label{eq:good-cosets-cover}
    \abs{\cup_{c^* \in C^*} (V+c^*)} \geq 0.05\alpha \abs{\F}^n
    \enspace.
\end{equation}

We are now ready to show that
\begin{equation*}
    \Pr_{\substack{a_1,a_2,a_3 \in \F^n \\ a_4 = v-a_1-a_2-a_3}}[a_1, a_2 \in A, a_3, a_4 \in -A] \geq \Omega(\alpha^5)
    \enspace.
\end{equation*}

\begin{proof}[Proof of \cref{thm:bogolyubov-quasipoly}]
    Fix $v \in V$.
    Since $\abs{D \cap (V+c^*)} \geq 0.8\abs{V}$ for every coset $V+c^*$ such that $c^* \in C^*$,
    it follows by the symmetry of $C^*$ that for every $c^* \in C^*$ we have
    at least $0.1 \cdot \abs{V}$ pairs $(u+c^*, v-u-c^*) \in D^2$ such that $u \in V$.
    Therefore, by \cref{eq:good-cosets-cover}
    for every $v \in V$ there are at least $0.1 \times 0.05 \alpha \abs{\F}^n$ different pairs
    $(u+c^*,v-u-c^*)$ such that both $u+c^* \in D \cap (V+c^*)$ and $v-u-c^* \in D \cap (V-c^*)$.

    Letting $d_1 = u+c^*$, so far we got that for every $v \in V$ we have
    $\Pr_{d_1 \in \F^n,d_2 = v - d_1}[d_1,d_2 \in D] \geq \Omega (\alpha)$.

    Recall that by the definition of $D$, every $d_1 \in D$ is a popular difference of elements of $A$,
    i.e. $\Pr_{\substack{a_1 \in \F^n \\ a_3 = d_1-a_1}}[a_1 \in A, a_3 \in -A] \geq \delta$.
    Similarly, for $d_2 \in D$ we have $\Pr_{\substack{a_2 \in \F^n \\ a_4 = d_2-a_2}}[a_2 \in A,a_4 \in -A] \geq \delta$.
    This implies that
    \begin{equation*}
        \Pr_{\substack{a_1,a_2,a_3 \in \F^n \\ a_4 = v-a_1-a_2-a_3}}[a_1, a_2 \in A, a_3, a_4 \in -A] \geq \Omega(\alpha \cdot \delta^2)
        \enspace,
    \end{equation*}
    as required.
\end{proof}

\medskip

We now turn to proving each of the two steps stated in \cref{lemma:sanders-step1} and \cref{lemma:sanders-step2}.

\subsection{Proof of Lemma~\ref{lemma:sanders-step1}}

The proof starts with the following lemma of Croot and Sisask~\cite{CS10}.

\begin{lemma}[Croot-Sisask {\cite[Proposition~3.3]{CS10}}]\label{lemma:croot-sisask}
    Let $A, B \seq \F^n$ be two sets, and let $\eps \in (0,1)$ and $p \geq 1$.
    Let $\alpha = \frac{\abs{A}}{\abs{\F}^n} \in (0,1)$.
    Then, there exists a set $X \seq \F^n$ of size $\abs{X} \geq (\alpha/2)^{O(p/\eps^2)} \abs{\F}^n$
    such that for all $x \in X$ it holds that
    \begin{equation*}
        \norm{\varphi_x * \varphi_A * 1_B - \varphi_A * 1_B}_p \leq \eps \,.
    \end{equation*}
\end{lemma}

Let $p = \log_2(1/\alpha)$, $t = \Theta(\log(1/\alpha))$, and $\eps = (1/40t)$.
By applying \cref{lemma:croot-sisask} we obtain a set $X \seq \F^n$ of size $\abs{X} \geq (\alpha/2)^{O(p/\eps^2)} \geq \alpha^{O(\log^3(1/\alpha))} \abs{\F}^n$.
We show below that $X$ satisfies \cref{eq:sanders-step1}.

Fix $x_1,\dots,x_t \in X$, and let $s = \sum_{i=1}^t x_i$.
Note first that by setting $B=D$ in \cref{lemma:croot-sisask} and combining it with triangle inequality
we get that
\begin{equation*}
    \norm{\varphi_{s} * \varphi_A * 1_D - \varphi_A * 1_D}_p \leq t \cdot \eps \leq 1/40
    \enspace,
\end{equation*}
where the last inequality is by the choice of $\eps = 1/40t$.
Let $q = p/(p-1)$, then by the choice of $p = \log_2(1/\alpha)$ we have
\begin{equation*}
    \norm{\varphi_A}_q = \left(\alpha \cdot \frac{1}{\alpha^q}\right)^{1/q} = \left(\frac{1}{\alpha}\right)^{1/p} \leq 2
    \enspace.
\end{equation*}
Then, by applying H\"older's inequality with $q = p/(p-1)$ we get
\begin{equation*}
    \abs{\ip{\varphi_s * \varphi_A * 1_D - \varphi_A * 1_D, \varphi_A}}
    \leq
    \norm{\varphi_s * \varphi_A * 1_D - \varphi_A * 1_D}_p \cdot \norm{\varphi_A}_q \leq 1/20
    \enspace.
\end{equation*}
By combining  the above inequality with \cref{eq:D-approx-A-A} we get
\begin{eqnarray*}
    \Pr_{a_1,a_2 \in A}[a_1 - a_2 - s \in D]
    & = & \ip{\varphi_s * \varphi_A * 1_D, \varphi_A} \\
    & = & \ip{\varphi_A * 1_D, \varphi_A} - \ip{\varphi_A * 1_D - \varphi_s * \varphi_A * 1_D, \varphi_A} \\
    & = & \ip{1_D, \varphi_A * \varphi_{-A}} - \ip{\varphi_A * 1_D - \varphi_s * \varphi_A * 1_D, \varphi_A} \\
    & \geq & (1-\delta/\alpha^2)  - 1/20 \geq 0.9
    \enspace,
\end{eqnarray*}
as required.

\subsection{Proof of Lemma~\ref{lemma:sanders-step2}}

The proof of this step is essentially Section 5 of \cite{Lovett15}.
The only (minor) difference is that we work over $\F_p$ and not over $\F_2$.

Given the set $X \seq \F^n$ from \cref{lemma:sanders-step1} 
of size $\abs{X} \geq \alpha^{O(\log^3(1/\alpha))} \abs{\F}^n$,
we define $\Spec_{\gamma}(X) = \{r \in \F^n : \abs{\widehat{\varphi}_X(r)} \geq \gamma \}$.
Since $\sum_{q \in \F^n} \abs{\widehat{\varphi}_X(r)}^2 = \E_z[\varphi_X(z)^2] = 1/\alpha$,
it follows that $\abs{\Spec_{\gamma}(X)} \leq \frac{1}{\alpha \gamma^2}$.
Chang's lemma provides a non-trivial bound on the dimension of the subspace containing $\Spec_{\gamma}(X)$.

\begin{lemma}[Chang~\cite{Chang02}]\label{lemma:chang}
    Let $X \seq \F^n$ of size $\abs{X} = \beta \cdot \abs{\F}^n$, and let $\gamma>0$. Then
    \begin{equation*}
        \dim(\Spec_\gamma(X)) \leq O \left(\frac{\log(1/\beta)}{\gamma^2} \right)
        \enspace.
    \end{equation*}
\end{lemma}

Define the subspace $V = \Spec_{1/2}(X)^\perp = \{v \in \F^n : \ip{v,r} = 0 \ \forall r \in \Spec_{1/2}(X)\}$.
\cref{lemma:chang} implies that $\dim(V) \geq n - O(\log^4(\abs{\F}^n/\abs{X})) \geq n - O(\log^4(1/\alpha))$.

We now show that $V$ indeed satisfies the guarantee of \cref{lemma:sanders-step2}.
For $x_1,x_2,\dots,x_t \in X$ let $s= \sum_{i=1}^t x_i$ as in the previous lemma.
By \cref{lemma:sanders-step1} for all $x_1,x_2,\dots,x_t \in X$ we have
\begin{equation*}\label{eq:2A+tX-in-D}
    \Pr_{a_1,a_2 \in A}[a_1 - a_2 - s \in D] \geq 0.9
    \enspace.
\end{equation*}
Next, we are comparing this probability to the following.
\begin{equation}\label{eq:V+2A+tX-in-D}
    \Pr_{\substack{v \in V \\ a_1,a_2 \in A}}[v + a_1 - a_2 - s\in D]
    \enspace.
\end{equation}
We claim that if we sample $v \in V$, $a_1,a_2 \in A$, and $x_1,\dots,x_t \in X$  uniformly at random
(and let $s = \sum_{i=1}^t x_i$), then the two quantities are close to each other. 
We prove this by rewriting the two probabilities using the Fourier expansion.
Note that the Fourier coefficients of $\varphi_V$ are simple to describe since $V$ is a linear subspace, and they are equal to $\widehat{\varphi_V}(r) = 1$ if $r \in V^\perp$ and $\widehat{\varphi_V}(r) = 0$ otherwise. Therefore, 
\begin{eqnarray*}
    \Pr_{\substack{v \in V \\ x_1,\dots, x_t \in X \\ a_1,a_2 \in A}}[v + a_1 - a_2 - s\in D]
    & = & \ip{\varphi_V *\varphi_A * \varphi_{-A} * \varphi_{-X}^{(t)} ,\1_D} \\
    & = & \sum_{r \in \F^n} \widehat{\varphi_V}(r) \cdot \widehat{\varphi_A}(r) \cdot \widehat{\varphi_{A}}(-r) \cdot \widehat{\varphi_X}^t(-r) \cdot  \overline{\widehat{\1_D}(r)} \\
    & = & \sum_{r \in V^\perp} \widehat{\varphi_A}(r) \cdot \widehat{\varphi_{A}}(-r) \cdot \widehat{\varphi_X}^t(-r) \cdot  \overline{\widehat{\1_D}(r)}
    \enspace.
\end{eqnarray*}
On the other hand
\begin{equation*}
    \Pr_{\substack{x_1,\dots, x_t \in X \\ a_1,a_2 \in A}}[a_1 - a_2 - s \in D]
    = \ip{\varphi_A * \varphi_{-A} * \varphi_{-X}^{(t)} ,\1}
= \sum_{r \in \F^n} \widehat{\varphi_A}(r)  \cdot \widehat{\varphi_A}(-r)  \cdot \widehat{\varphi_X}^t(-r) \cdot \overline{\widehat{\1_D}(r)}
    \enspace.
\end{equation*}
This implies
\begin{eqnarray*}
    \abs{\Pr_{\substack{v \in V \\ x_1,\dots, x_t \in X \\ a_1,a_2 \in A}}[v + a_1 - a_2 - s\in D]
    - \Pr_{\substack{x_1,\dots, x_t \in X \\ a_1,a_2 \in A}}[a_1 - a_2 - s \in D]}
    & = &
    \sum_{r \not\in V^\perp} \widehat{\varphi_A}(r) \cdot \widehat{\varphi_A}(-r) \cdot \widehat{\varphi_X}^t(r) \cdot \overline{\widehat{\1_D}(r)}  \\
    & \leq &
    \sum_{r \not\in V^\perp} \abs{ \widehat{\varphi_A}(r) \cdot \widehat{\varphi_A}(-r)  \cdot 2^{-t} \cdot \overline{\widehat{\1_D}(r)} } \\
    & \leq &
    2^{-t} \sum_{r \in \F^n} \abs{ \widehat{\varphi_A}(r)  \cdot \widehat{\varphi_A}(-r) } \\
    \mbox{By Cauchy-Schwarz}& \leq &
    2^{-t} \sum_{r \in \F^n} \abs{ \widehat{\varphi_A}^2(r)} \\
    & = & 2^{-t} \cdot \E_{x \in \F^n}[\varphi_A^2(x)] \\
    & = & \frac{1}{\alpha \cdot 2^t} < 0.05
    \enspace,
\end{eqnarray*}
where the last inequality holds due to the choice of $t = O(\log(1/\alpha))$. Therefore, 
\begin{equation*}
    \Pr_{\substack{v \in V \\ x_1,\dots, x_t \in X \\ a_1,a_2 \in A}}[v + a_1 - a_2 - s\in D]
    \geq \Pr_{\substack{x_1,\dots, x_t \in X \\ a_1,a_2 \in A}}[a_1 - a_2 - s \in D] - 0.05 \geq 0.85
    \enspace.
\end{equation*}
Finally, we can fix $a_2 + s$ maximizing the probability, and let $b = -a_2-s$ to conclude the proof of \cref{lemma:sanders-step2}.
\subsection{\texorpdfstring{Algorithmic construction of the subspace $V$}{Algorithmic construction of the subspace~V}}
Given a set $A$ we can construct $V$ using the algorithm described in \cite{BenSassonRZTW14}.
Indeed, the only ingredients we need for our construction are the set $X$ from \cref{lemma:sanders-step1} and the subspace $V$ guaranteed by \cref{lemma:sanders-step2}.

A straightforward inspection of the algorithm described in \cite{BenSassonRZTW14} gives the desired result.
Informally, the algorithm works as follows:
The set $X$ is defined as the set of all $x \in X$ such that
\begin{equation*}
    \norm{\varphi_x * \varphi_A * 1_D - \varphi_A * 1_D}_p \leq \eps
    \enspace.
\end{equation*}
Note that given $x \in \F^n$ we can estimate the norm efficiently (up to a small error).
This gives us a membership oracle to the set $X$.

Given such a query oracle, we can use Goldreich-Levin algorithm over $\F$ to compute $R = \Spec_{1/2}(X)$~\cite{Akavia08}, or more precisely its superset that is not too large, and using it we define the subspace $V = \{v \in \F^n : \ip{v,r} = 0 \; \forall r \in R\}$.
\end{document}